\documentclass[colorinlistoftodos,runningheads]{llncs}

\usepackage{syntax}
\usepackage{amsmath}
\usepackage{amssymb}
\usepackage{mathtools}
\usepackage{bussproofs}
\usepackage[table]{xcolor}
\usepackage{todonotes}
\usepackage{listings}
\usepackage{tikz,tikzsymbols}
\usepackage[inline]{enumitem}
\usepackage{abs}
\usepackage{jmllistings}
\usepackage{soul}
\usepackage[hidelinks]{hyperref}
\usepackage{tikz}
\usetikzlibrary{tikzmark}
\newcommand{\LVar}{\ensuremath{\mathsf{LVar}}}

\newcommand{\RVar}{\ensuremath{\mathsf{RecVar}}}
\newcommand{\stateFml}[1]{\lceil#1\rceil}
\newcommand{\chop}{\ensuremath*\!*}
\newcommand{\concat}{\ensuremath\cdot}

\newcommand{\muFml}[2]{\ensuremath{\mu#1.#2}}

\newcommand{\domain}{\ensuremath{\mathsf{D}}}
\newcommand{\interp}{\ensuremath{\mathsf{I}}}

\newcommand{\recVarAsgn}{\ensuremath{\rho}}

\newcommand{\upP}[1]{\upl#1\upr}
\newcommand{\eval}[2]{\ensuremath{[\![#1]\!]_{#2}}}
\newcommand{\evalB}[2]{\ensuremath{[\![#1]\!]_{#2,\beta}}}

\newcommand{\evalMu}[6]{\ensuremath{[\![#1]\!]_{#2,#3,#4,#5,#6}}}
\newcommand{\evalMuS}[1]{\ensuremath{\evalMu{#1}{\domain}{\interp}{\beta}{\sigma}{\recVarAsgn}}}
\newcommand{\evalPhi}[4]{\ensuremath{[\![#1]\!]^{#2}_{#3,#4}}}
\newcommand{\evalPhiS}[1]{\ensuremath{\evalPhi{#1}{
  }{\beta}{\rho}}}
\newcommand{\evalPhiB}[1]{\ensuremath{[\![#1]\!]_{\beta}}}
\newcommand{\evalPhiNoArg}[1]{\ensuremath{[\![#1]\!]}}

\newcommand{\evalPrg}[4]{\ensuremath{[\![#1]\!]_{#2,#3,#4}}}
\newcommand{\evalPrgCtx}[4]{\ensuremath{[\![#1]\!]_{#2,#3,#4}}^{ctx}}
\newcommand{\evalPrgS}[1]{\ensuremath{\evalPrg{#1}{\domain}{\interp}{\sigma}}}
\newcommand{\evalPrgSCtx}[1]{\ensuremath{\evalPrgCtx{#1}{\domain}{\interp}{\sigma}}}

\newcommand{\singleton}[1]{\ensuremath{\langle#1\rangle}}
\newcommand{\chopT}{\ensuremath{\underline{\chop}}}
\newcommand{\States}{\ensuremath{\Sigma}}
\newcommand{\Traces}{\ensuremath{\mathsf{Traces}}}

\newcommand{\noEvent}[1]{\stackrel{#1}{\ensuremath{\cdot\cdot}}}

\newcommand{\Todo}[2][]{%
\ifthenelse{\equal{#1}{done}}
  {
    \todo[backgroundcolor=green!30,inline]{#2}
  }
    {
      \ifthenelse{\equal{#1}{warning}}
      {
        \todo[backgroundcolor=orange!30,inline]{#2}
      }
      {
        {
          \ifthenelse{\equal{#1}{ongoing}}
          {
            \todo[backgroundcolor=blue!30,inline]{#2}
          }
          {
            \todo[backgroundcolor=red!30,inline]{#2}
          }
        }
      }
    }
}

\newcommand{\pushEv}{\ensuremath{\mathsf{pushEv}}\xspace}
\newcommand{\popEv}{\ensuremath{\mathsf{popEv}}\xspace}
\newcommand{\callEv}{\ensuremath{\mathsf{callEv}}\xspace}
\newcommand{\retEv}{\ensuremath{\mathsf{retEv}}\xspace}
\newcommand{\startEv}{\ensuremath{\mathsf{startEv}}\xspace}
\newcommand{\finishEv}{\ensuremath{\mathsf{finishEv}}\xspace}

\newcommand{\upl}{\ensuremath{\{}}
\newcommand{\upr}{\ensuremath{\}}}
\newcommand{\update}{\ensuremath{\mathcal{U}}}

\newcommand{\judge}[2]{\ensuremath{#1:#2}}
\newcommand{\finiteNoM}[1]{\ensuremath{\overset{\begin{subarray}{l}#1\end{subarray}}}{\cdot\cdot}}
\newcommand{\finiteNoEv}[1]{\ensuremath{\overset{\begin{subarray}{l}#1\end{subarray}}}{\cdot\cdot\cdot}}

\newcommand{\turnstyle}{\vdash}

\newcommand{\ruleName}[1]{{\ensuremath{\mathsf{#1}}}}

\newcommand{\seq}[2]{#1\turnstyle#2}
\newcommand{\semseq}[2]{#1\models#2}

\newcommand{\sequent}[2]{\seq{\Gamma\ifx#1\relax\else,#1\fi} 
                             {#2}}
\newcommand{\semsequent}[2]{\semseq{\Gamma\ifx#1\relax\else,#1\fi} 
                                   {#2}}

\newcommand{\adfrac}[2]{\genfrac{}{}{}{0}
  {\begin{array}{l}#1\end{array}}
  {\begin{array}{l}#2\end{array}}}

\newcommand{\seqRule}[3]{\mbox{\small\ruleName{#1}}\ 
  \adfrac{#2}{#3}}

\newcommand{\lastEv}[1]{lastEv(#1)}

\newcommand{\Simple}[1]{\mbox{\lstinline{#1}}}
\newcommand{\code}[1]{\mbox{\lstinline|#1|}}

\newcommand{\Red}[1]{{\color{red}#1}}
\newcommand{\Blue}[1]{{\color{blue}#1}}

\newcommand{\mycomment}[1]{}

\newcommand{\methods}[1]{\method(#1)}
\newcommand{\Gi}{\ensuremath \mathcal{G}}

\newcommand{\scope}{\ensuremath \mathit{sc}}

\newcommand{\last}{\ensuremath \mathrm{last}}
\newcommand{\first}{\ensuremath \mathrm{first}}
\newcommand{\cont}[1]{\ensuremath \mathrm{K}(#1)}

\newcommand{\contP}[2]{\ensuremath \mathrm{K}^{#1}(#2)}
\newcommand{\contF}[2]{\ensuremath \mathrm{K}^{#1}(#2)}
\newcommand{\zero}{\ensuremath \circ}
\renewcommand{\zero}{\mbox{\bottle}}

\newcommand{\subst}[2]{[#1\leftarrow#2]}

\newcommand{\trueSem}{\ensuremath \mathrm{t\!t}}
\newcommand{\falseSem}{\ensuremath \mathrm{f\!f}}

\newcommand{\valB}[2]{\ensuremath \mathrm{val}_{#1}(#2)}

\newcommand{\valP}[2]{\ensuremath \mathrm{val}_{\sigma}^{#1}(#2)}
\newcommand{\valBP}[3]{\ensuremath \mathrm{val}_{#1}^{#2}(#3)}

\newcommand{\valDnoargs}{\ensuremath \mathrm{val}_{\sigma}}
\newcommand{\chopSem}{\underline{\mathbin{\ast\ast}}}
\newcommand{\updatestate}[3]{#1[#2 \mapsto #3]}
\newcommand{\chopTrSem}[2]{\ensuremath #1 \chopSem #2}
\newcommand{\concatTr}[2]{\ensuremath #1\cdot #2}
\newcommand{\consTr}[2]{\ensuremath #1 \curvearrowright #2}
\newcommand{\cons}{\ensuremath \curvearrowright}

\newcommand{\UNUSED}[1]{}
\newcommand{\OUTDATED}[1]{}

\DeclareMathOperator{\method}{procedures}

\DeclareMathOperator{\lp}{lookup}

\newcommand{\lookup}[1]{\lp(#1)}

\newcommand{\ifStmt}[2]{\Simple{if} \ #1 \ \{~ #2 ~\}\xspace}
\newcommand{\whileStmt}[2]{\Simple{while} \ #1 \ \{~ #2 ~\}\xspace}
\newcommand{\assignStmt}[2]{#1=#2\xspace}
\newcommand{\skipStmt}{\Simple{skip}\xspace}
\newcommand{\returnStmt}[1]{\Simple{return}\ #1\xspace}

\newcommand{\semPair}[2]{#1,\,K(#2)}

\newcommand{\semSeq}[3]{#1 \overset{#3}{\to} #2}

\newcommand{\Main}{\Simple{main}}

\newcommand{\inline}{\mathrm{inline}}

\newcommand{\retEvP}[2]{\mathsf{retEv}_{#1}(#2)} 
\newcommand{\callEvP}[2]{\mathsf{callEv}_{#1}(#2)}
\newcommand{\pushEvP}[2]{\mathsf{pushEv}_{#1}(#2)}
\newcommand{\popEvP}[2]{\mathsf{popEv}_{#1}(#2)}
\newcommand{\expose}[1]{\mathsf{expose}(#1)} 

\newcommand{\evTrioPV}[3]{\textit{ev}_{#1}^{#3}({#2})} 

\newcommand{\ev}[1]{\textit{ev}(#1)} 
\newcommand{\stateorevent}{t}

\newcommand{\num}[2]{\ensuremath \#_{#1}({#2})}
\newcommand{\many}[1]{\overline{#1}}

\newcommand{\RuleL}[1]{\LeftLabel{$\mathsf{#1}$}}

\newcommand{\gcondrule}[4]{ 
  \begin{equation}
    \label{eq:rule#1}
    \textsc{({#2})}\,
  \begin{array}{c} 
    #3 \\[1pt] 
    \hline\\[-7pt]
    #4 
  \end{array} 
 \end{equation}
    }

\title{Towards Trace-based Deductive Verification (Tech Report)}

\author{Richard Bubel\inst{1} \and Dilian Gurov\inst{2} \and Reiner
 H\"ahnle\inst{1} \and Marco Scaletta\inst{1}}

\institute{Technical University of Darmstadt, Germany,
  \email{<firstName>.<lastName>@tu-darmstadt.de} \and KTH Royal
  Institute of Technology, Sweden, \email{dilian@kth.se}}

\begin{document}

\maketitle
\vspace*{-1cm}


\begin{abstract}
  Contracts specifying a procedure's behavior in terms of pre- and
  postconditions are essential for scalable software verification, 
  but cannot express any constraints on the events occurring during 
  execution of the procedure. This
  necessitates to annotate code with intermediate assertions,
  preventing full specification abstraction.
  We propose a logic over symbolic traces able to specify recursive
  procedures in a modular manner that refers to specified programs
  only in terms of events. We also provide a deduction system based on
  symbolic execution and induction that we prove to be sound relative
  to a trace semantics.
  Our work generalizes contract-based to trace-based deductive
  verification.
\end{abstract}




\section{Introduction}


To make deductive verification scale, modular specification and
verification is of essence~\cite{HaehnleHuisman19}. In imperative
programming languages modularity manifests itself at the granularity
of procedure calls. A procedure's behavior is specified in terms of a
\emph{contract} \cite{Liskov79,Meyer92}: a pair of first-order pre-
and postconditions. During verification each procedure call in the code
is replaced with a check that the contract's precondition holds at
this point, the call's effect is approximated by assuming the
postcondition. In consequence, verification of the called code is
replaced (or approximated) by first-order constraints. The overall
verification effort for a program is proportional to its length,
instead of the unfolded (and unbounded) call graph.  The
contract-based approach works for recursive procedures
\cite{Hoare71,Oheimb99}. It requires an induction principle
\cite{BrotherstonSimpson11,LidstromGurov21} and is realized in most
state-of-art deductive verification systems, such as
\cite{ABBHS16,Frama-C15,LeinoWuestholz14}.

The fundamental limitation of pre-/postcondition contracts, which in
the following are called \emph{state-based} contracts, is the
inability to specify events happening \emph{during} execution of a
procedure. This is not only problematic for specifying concurrent
programs, but already an issue for interactive programs or loop
specification: It is necessary to annotate code with intermediate
assertions, which worsens readability and impedes full
abstraction from the implementation during specification. Indeed,
contracts written in most specification languages in deductive
verification \cite{acsl,JML-Ref-Manual} cannot be used independently 
of the specified code.

Lack of the ability to create abstract specifications motivates the
design of a logic serving as an abstract, \emph{trace-based} contract
language that is sufficiently expressive to model recursive calls (and
loops). The central idea is to represent the control structures
embodied in the control flow graph (CFG) of a program.  From the CFG
view two requirements on a trace-based contract language can be
derived:
\begin{enumerate*}[label=(\roman*)]
\item it must be possible to specify scopes relating to procedure
  calls and other control structures: to this end we employ
  \emph{events} that are semantically connected with the specified
  code;
\item one must represent cyclic edges in the CFG in a finite manner:
  this is achieved by a least fixed-point operator.
\end{enumerate*}
Both aspects are not new, of course.  We draw from ideas first
expressed in interval temporal logic \cite{HalpernShoham91} and the
modal $\mu$-calculus~\cite{kozen-83-muc}, respectively.

However, we use the resulting logic of \emph{trace contracts}
differently from temporal/modal logic: like in Hoare logic
\cite{Hoare69,Hoare71} and dynamic logic \cite{HKT00,Pratt76}, we
admit explicit \emph{judgments} about a concrete program: if $s$ is a
program and $\Phi$ a trace contract, then $\judge{s}{\Phi}$ expresses
that any execution of $s$ results in one of the traces characterized
by $\Phi$. As mentioned above, we incorporate \emph{events} into
contracts to enable \emph{full abstraction} from the specified code
(cf.\ \cite{JeffreyRathke05} who use \emph{actions} to
design an abstract semantics).
Consequently, our logic generalizes Hoare/dynamic logic from
state-based to trace-based contracts via events (which permit scoping)
and recursive specifications. This sounds more complex than it
is. Still, to keep the presentation intuitive and manageable, it is
crucial to choose
\begin{enumerate*}[label=(\roman*)]
\item a suitable deductive framework and
\item an adequate semantics to demonstrate soundness.
\end{enumerate*}

Regarding the first point, we adapt a symbolic execution calculus for
deductive verification in dynamic logic~\cite{ABBHS16}: it reduces a
program $s$ via symbolic execution to a sequence of elementary
symbolic state updates $\mathcal{U}_s$ (plus path conditions), hence,
it reduces judgments of the form $\judge{s}{\Phi}$ to
$\judge{\mathcal{U}_s}{\Phi}$. Adding events to this formalism is
easy. 
Dynamic logic is closed under propositional and first-order operators,
which enables formulation of a fixed-point induction rule over the
syntactic structure of programs.
Concerning the second issue, an adequate semantics must be
\emph{trace-based}, i.e.\ it provides
\begin{enumerate*}[label=(\alph*)]
\item for a given state~$\sigma$ and program $s$ the trace of $s$ when
  started in $\sigma$ and
\item for a given trace formula $\Phi$ the set of all traces it
  characterizes.
\end{enumerate*}
Since the deduction rules resemble symbolic execution of a
program $s;s'$, where $s$ is the leading statement and $s'$ the
remaining code, an adequate semantics is defined \emph{locally} for
each kind of statement $s$ and continuation $s'$. A general local,
trace-based semantic framework was suggested in
\cite{LAGCArxiv22,DHJPT17}---here, we specialize it to sequential
programs with recursive procedures.

To summarize, the main contributions of this paper are:
\begin{enumerate*}[label=(\arabic*)]
\item A trace-based specification language that permits fully
  abstract, modular specification of recursive procedures;
\item the extension of a symbolic execution calculus for
  procedure-modular verification by structural induction and trace
  abstraction;
\item a soundness proof of the deduction rules based on a local trace
  semantics.
\end{enumerate*}

In Sect.~\ref{sec:local-trace-semantics} we present the local trace
semantics, then define syntax and semantics of trace contracts in
Sect.~\ref{sec:logic-trace-contracts}. The deductive verification
system is given in Sect.~\ref{sec:calculus} (for space reasons, the 
soundness proofs are in the appendix).
Related work is in Sect.~\ref{sec:related-work}, conclusion with
future work in Sect.~\ref{sec:concl-future-work}.



\section{Local Trace Semantics of Recursive Procedures}
\label{sec:local-trace-semantics}

We provide a LAGC-style~\cite{LAGCArxiv22,DHJPT17} trace semantics for
a sequential language with recursive procedures. Its grammar is shown in Fig.~\ref{fig:lang-syntax}.

\begin{figure}  
\begin{center}
$\begin{array}[t]{r@{\hspace{2pt}}r@{\hspace{2pt}}l@{\qquad\qquad\quad}r@{\hspace{2pt}}r@{\hspace{2pt}}l}
P\in\mathit{Prog}&::=&\overline{M}\,\{d \ s\} &
M\in\mathit{ProcDecl} & ::= & m(x)\{\scope\} \\
d\in\mathit{VarDecl}  & ::= &\varepsilon \mid x; \ d &
\scope\in Scope &::= & \{d \ s; \ \returnStmt{e}\}\\
s\in\mathit{Stmt} &::= & \multicolumn{4}{l}{\skipStmt \mid  x = e \mid x = m(e) \mid \ifStmt{e}{s} \mid s;s \mid \whileStmt{e}{s}}\\
\end{array}
$\smallskip

\noindent Global lookup table:
$\Gi= \{\langle \many{m(x)\ \scope} \rangle\mid  m\in \methods{P}\}\ $
\end{center}
\caption{Syntax of an imperative language with recursive procedures}
\label{fig:lang-syntax}
\end{figure}

The definition of integer and boolean
expressions is standard, they are assumed to be well-typed. Scopes of
procedure bodies may only write to local integer variables, initialized
to zero, i.e.\ procedure calls have no side effects. We stress that
this is not a fundamental limitation, but to deal with side effects
and aliasing~\cite{ABBHS16,ParkinsonSummers12} is orthogonal to the
goals of our paper and the restriction greatly simplifies the
technical issues.

\begin{example}\label{lst:running-example}\hfill\\
\begin{minipage}[t]{0.54\textwidth}
  We illustrate the concepts in this paper with the running example
  shown on the right.  The behavior of procedure $\code{m}$ is the identity
  function for input \code{k}, but the result is computed by
  \code{k} (non tail-)recursive calls.
\end{minipage}
\hfill%
\begin{minipage}[t]{0.45\textwidth}
\begin{lstlisting}[mathescape=true,aboveskip=-2em,belowskip=0em]{}
  m(k) {
    r; // initialized to $0$
    if (k != 0)  
      { r = m(k-1); r = r + 1 };
    return r
  }
\end{lstlisting}%
\end{minipage}
\end{example}



\begin{definition}[State, Update]
  Let $\mathit{Var}$ be a set of program variables and $\mathit{Val}$ a set of values,
  with typical elements $x$ and $v$, respectively.  A \emph{state}
  $\sigma\in\Sigma$ is a partial mapping
  $\sigma: \mathit{Var} \rightharpoonup \mathit{Val}$
  from variables to values. The notation $\sigma[x\mapsto v]$
  expresses the \emph{update} of state $\sigma$ at $x$ with value $v$
  and is defined as $\sigma[x\mapsto v](y)=v$ if $x=y$ and
  $\sigma[x\mapsto v](y)=\sigma(y)$ otherwise.



\end{definition}

There is a standard evaluation function $\valDnoargs$ for expressions,
for example, in a state $\sigma=[x\mapsto0, y\mapsto1]$ we have
$\valB{\sigma}{x + y}=\valB{\sigma}{x} + \valB{\sigma}{y}=0+1=1$.

\begin{definition}[Context] 
  A \emph{(call) context} is a pair $ctx=(m,id)$, where $m$ is a
  \emph{procedure name} and $id$ is a \emph{call identifier}.  Given a
  call identifier $id$ we denote with $\emph{\code{res}}_{id}$ the unique name of
  the return variable associated with the call.
\end{definition}

We define events sufficient to characterize the scope and call
structure of recursive procedures, but one could define further event
types, for example, input events or, in a concurrent setting,
suspension events.

\begin{definition}[Event Marker]\label{def:event} 
  Let $m$ be a procedure name, $e$ a parameter, $id$ a call
  identifier, and $v$ a return value. Then $\callEvP{}{m,e,id}$ and
  $\retEvP{}{v}$ are \emph{event markers} associated with a procedure
  call and a return statement, respectively.  We also introduce event
  markers associated with the start and end of a computation in
  context $ctx$, defined as $\pushEvP{}{ctx}$ and $\popEvP{}{ctx}$,
  respectively.  We denote with $\ev{\many{e}}$ a generic event marker
  over expressions $\many{e}$.
\end{definition}

\begin{definition}[Trace] 
A \emph{trace} $\tau$ is defined by the following rules
(where $\varepsilon$ denotes the empty trace):
\qquad $
  \begin{array}{l@{\;::=\;}l@{\qquad}l@{\;::=\;}l}
    \tau & \varepsilon~|~\consTr{\tau}{t} &
    \stateorevent & \sigma~|~\ev{\many{e}}
  \end{array}
  $
\end{definition}


This definition declares traces as sequences over events and states,
but we need to uniquely associate an event $\ev{\many{e}}$ with a
state $\sigma$. This is done by inserting event $\ev{\many{e}}$ into a
trace between two copies of $\sigma$.\footnote{Alternatively, one
  could use state transitions labeled with (possibly empty) events.}
We define the \emph{event trace} $\evTrioPV{\sigma}{\many{e}}{}$ as
$ \evTrioPV{\sigma}{\many{e}}{} =
\consTr{\langle\sigma\rangle}{\consTr{\ev{\valB{\sigma}{\many{e}}}}{\sigma}}
$.  Events do not change a state.


Sequential composition ``\lstinline|r;s|'' of statements is
semantically modeled as trace composition, where the trace from
executing \lstinline|r| ends in a state from which the execution trace
of \lstinline|s| proceeds.  Thus the trace of \lstinline|r| ends in
the same state as where the trace of \lstinline|s| begins.  This
motivates the \emph{semantic chop} ``$\chopSem$'' on
traces~\cite{HalpernShoham91,HarelKP80,NakataUustalu15} that we often
use, instead of the standard concatenation operator ``$\concatTr{}{}$''.


\begin{definition}[Semantic Chop on Traces] 
  Let $\tau_1,\,\tau_2$ be traces, assume $\tau_1$ is non-empty and
  finite. The \emph{semantic chop} $\tau_1 \chopSem \tau_2$ is defined
  as
  $ \chopTrSem{\tau_1}{\tau_2} = \concatTr{\tau}{\tau_2} $, where
  $\tau_1=\consTr{\tau}{\sigma}$,
  $\tau_2=\concatTr{\langle\sigma'\rangle}{\tau'}$, and
  $\sigma=\sigma'$.
  When $\sigma\neq\sigma'$ the result is undefined.
  We overload the semantic chop symbol for sets of traces: 
  $$T_1 \chopSem T_2 = \{\tau_1 \chopSem \tau_2 \mid \tau_1 \in T_1,\,
  \tau_2 \in T_1,\,\last(\tau_1)=\first(\tau_2)\}\enspace.$$
\end{definition}

\begin{example}
  Let $\tau_1=\singleton{\sigma}\cons\sigma[\code{x}\mapsto1]$ and
  $\tau_2=\singleton{\sigma[\code{x}\mapsto1]}\cons\sigma[\code{x}\mapsto1,
  \code{y}\mapsto2]$, then $\tau_1 \chopSem \tau_2$ =
  $\singleton{\sigma}\cons\sigma[\code{x}\mapsto1]\cons\sigma[\code{x}\mapsto1,
  \code{y}\mapsto2]$.
\end{example}


Our language is deterministic, so we design local evaluation
$\valP{}{s}$ of a statement $s$ in state $\sigma$ to return a single
trace: The result of $\valP{}{s}$ is of the form
$\concatTr{\tau}{\contF{}{s'}}$, where $\tau$ is an initial
(small-step) trace of $s$ and $\contF{}{s'}$ contains the remaining,
possibly empty, statement $s'$ yet to be evaluated.

\begin{definition}[Continuation Marker] 
  Let $s$ be a program statement, then $\contP{}{s}$ is a
  \emph{continuation marker}. The empty continuation is denoted with
  $\contP{}{\zero}$ and expresses that nothing remains to be
  evaluated.
\end{definition}

The local evaluation rules defining $\valP{}{s}$ are in
Fig.~\ref{fig:local}. The rules for call and return emit suitable
events, the rule for sequential composition assumes empty leading
continuations are discarded, the remaining rules are straightforward.

\begin{figure}[t]
  \begin{align*}
    \valP{}{\Simple{skip}} =& \  
                              \concatTr{\langle\sigma\rangle}{\contF{}{\zero}} \qquad\qquad
    \valP{}{\assignStmt{x}{e}}  = 
                        \concatTr{\consTr{\langle\sigma\rangle}{\updatestate{\sigma}{x}{\valP{}{e}}}}{\contF{}{\zero}}\ \\ \valP{}{\{\  s\ \}} =& \  \valP{}{s} \qquad\qquad\quad
                        \valP{}{\Simple{return} \ e}
                        = \ \concatTr{\retEvP{\sigma}{\valP{}{e}}}{\contF{}{\zero}}\\
    \valP{}{\Simple{if}~e~\{ ~s~\}}  = & \begin{cases}
      \concatTr{\langle\sigma\rangle}{\contF{}{s}}, \ \text{if} \ \valBP{\sigma}{}{e} = \trueSem\\
      \concatTr{\langle\sigma\rangle}{\contF{}{\zero}}\, \ \text{otherwise}
    \end{cases}\\
    \valP{}{\Simple{while}~e~\{~s~\}} =& \ \valP{}{\Simple{if}~e~\{
                                         ~s;~\Simple{while}~e~\{~s~\}\}}\\
    \valP{}{r;s}
    = & \ \concatTr{\tau}{\contF{}{r'; s}}, \ \text{where} \ \valP{}{r}=\concatTr{\tau}{\contF{}{r'}}\ \text{and} \ \zero; s \rightsquigarrow s \notag \\
    \valP{}{\{\ x;d\ s\ \}} =& \   
                               \concatTr{\consTr{\langle\sigma\rangle}{\sigma[x' \mapsto 0]}}{\contF{}{\{ d\ s\subst{x}{x'}\ \}}}\ \text{with} \ x' \not \in dom(\sigma) \notag \\ 
   \valP{}{\assignStmt{x}{m(e)}} =& \ \concatTr{\callEvP{\sigma}{m,\valP{}{e},id}}{\contF{}{\assignStmt{x}{\code{res}_{id}}}}\ \text{where} \ \code{res}_{id} \not\in dom(\sigma) \notag
\end{align*}
\caption{Local Program Semantics}
\label{fig:local}
\end{figure}

We define schematic traces that allow us to succinctly characterize
sets of traces (not) containing certain events via matching. The
notation $\finiteNoEv{\many{ev}}$ represents the set of all non-empty,
finite traces \emph{without} events of type $ev\in\many{ev}$.  Symbol
$\finiteNoEv{}$ is shorthand for $\finiteNoEv{\emptyset}$ and $Ev$
includes all event types in Def.~\ref{def:event}.
With $\tau_1 \finiteNoEv{\many{ev}} \tau_3$ we denote 
the set of well defined traces $\tau_1 \chopSem \tau_2 \chopSem \tau_3$
so that $\tau_2 \in \finiteNoEv{\many{ev}}$.
%
Schematic traces make it easy to retrieve the
most recent event and the current call context from a given trace:

\begin{definition}[Last Event, Current Context] 
  \label{def:lastev-currctx}
  Let $\tau$ be a non-empty trace.
  \[\lastEv{\tau}=\begin{cases}
    ev_\sigma(\many{v})& \tau\in\finiteNoEv{} \, {ev_\sigma(\many{v})} \,  \finiteNoEv{Ev}\\
    NoEvent& \text{otherwise}
  \end{cases}\]
  \vspace{-5pt}
\[currCtx(\tau) = 
\begin{cases}
  \vspace{-10pt}
  ctx& \tau\in\,\finiteNoEv{} \,\pushEvP{\sigma}{ctx} \finiteNoEv{\pushEv,\popEv}\\
  currCtx(\tau')& \tau\in\tau' \chopSem \pushEvP{\sigma}{ctx}  \finiteNoEv{} \popEvP{\sigma'}{ctx}\finiteNoEv{\pushEv\\\popEv}\\
  (\Main,nul)& otherwise
\end{cases}\]
\end{definition}


Local evaluation of a statement $s$ yields a small step $\tau$ of $s$
plus a continuation~$\contF{}{s'}$. Therefore, traces can be extended
by evaluating the continuation and stitching the result to
$\tau$. This is performed by \emph{composition rules} that operate on
a configuration of the form $\tau,\cont{s}$. The process terminates
when all statements are evaluated, i.e.\ $\cont{s}=\cont{\zero}$. In
this case $\tau$ is the semantics resulting from evaluation of a
program.
There are three composition rules. The following rule evaluates a
statement that has not been directly preceded by a call or return and
extends the current trace accordingly.

\gcondrule{progress-rule}{Progress}{
  \tau \not\in \finiteNoEv{}\,\callEvP{\sigma}{\_,\_,\_} \qquad\qquad \tau \not\in \finiteNoEv{}\,\retEvP{\sigma}{\_} \\
  \sigma = \last(\tau)  \qquad\qquad 
  \valP{}{s} = \concatTr{\tau'}{\cont{s'}}
}{
  \tau, \cont{s} \to  \chopTrSem{\tau}{\tau'}, \cont{s'}
}

Procedure calls and returns must be handled differently to model the
change of call context and the return of the computed result.  Right
after a call statement was evaluated, i.e.\ when $\tau$ ends with a
call event, the call context is switched to the new context $ctx$ and
the body $s'$ of the called procedure is inlined:
\gcondrule{call-rule}{Call}{
  ctx=(m,id) \qquad\qquad
  \tau \in \finiteNoEv{}\,\callEvP{\sigma}{m,v,id}\\
  s' = sc \subst{x}{v},\ \text{where} \lookup{m, \Gi} = m(x) \ \scope
}{
  \tau, \cont{s} \to  \chopTrSem{\tau}{\pushEvP{\sigma}{ctx}}, \cont{s';s}
}

Immediately after a return statement the returned value is assigned to
the \emph{result variable} associated with the current context and the
context is switched back to the old context $ctx$ retrieved from
$\tau$ via matching. Together,
rules~\eqref{eq:rulecall-rule}--\eqref{eq:ruleret-rule} model
synchronous semantics of procedure calls.
\gcondrule{ret-rule}{Return}{
  ctx=currCtx(\tau)=(m,id)\qquad\qquad
  \tau \in \finiteNoEv{}\,{\retEvP{\sigma}{v}}
}{
  \tau, \cont{s} \to  \consTr{\tau}{\sigma [\code{res}_{id} \mapsto v] \chopSem \popEvP{\sigma [\code{res}_{id} \mapsto v]}{ctx}}, \cont{s}
}  

\begin{example}\label{example:composition-rules}
  We evaluate configuration
  $\singleton{\sigma_0}, \cont{\assignStmt{\code{x}}{\code{m(1)}}}$ for empty
  $\sigma_0$, with $\code{m}$ as defined in
  Expl.~\ref{lst:running-example}. Let $s=mb\subst{\code{k}}{1}$, where $mb$
  is the body of $\code{m}$.  Applying the progress and call rule we obtain the initial rule sequence:
  \begin{align*} 
    &\singleton{\sigma_0}, \cont{\assignStmt{\code{x}}{\code{m(1)}}}\to\callEvP{\sigma_0}{\code{m},1,0},\cont{\assignStmt{\code{x}}{\code{res}_0}} \\
    &\to \callEvP{\sigma_0}{\code{m},1,0} \chopSem \pushEvP{\sigma_0}{(\code{m},0)},\cont{\code{s};\assignStmt{\code{x}}{\code{res}_0}}
  \end{align*}
\end{example}

It is straightforward to see that the local evaluation and the
composition rules are exhaustive and deterministic, which justifies
the following definition:

\begin{definition}[Program Trace] 
  Given a program $s$ and a state $\sigma$, the \emph{trace} of
  $s$ (with implicit lookup table) is the maximal sequence obtained by
  repeated application of composition rules, starting from
  $\singleton{\sigma}, \cont{s}$. If finite, it has the form
  $\singleton{\sigma}, \cont{s} \to \cdots \to \tau, \cont{\zero}$,
  also written
  $\singleton{\sigma}, \cont{s} \overset{*}{\to} \tau, \cont{\zero}$.
\end{definition}


\begin{definition}[Program Semantics] 
  \label{def:glob-sem}
  The \emph{semantics} of a program $s$ is
  only defined for terminating programs as $\eval{s}{\tau}=\tau'$ if
  $\tau, \cont{s} \overset{*}{\to} \tau \chopSem \tau', \cont{\zero}$.
  
\end{definition}
From this definition and the local semantics of sequential statements
it is easy to prove Proposition~\ref{prop:glob-sem-comp} in the
appendix by straightforward induction.

\begin{example}[Continuing Example~\ref{example:composition-rules}]
  Fig.~\ref{fig:big-step-semantics-example} visualizes the semantics
  $\eval{\assignStmt{\code{x}}{\code{m}(1)}}{\singleton{\sigma_0}}$
  with the intermediate states
  $\sigma_1= \sigma_0[\code{r}'\mapsto 0]$,
  $\sigma_2= \sigma_1[\code{r}''\mapsto 0]$,
  $\sigma_3= \sigma_2[\code{res}_1\mapsto 0]$,
  $\sigma_4= \sigma_3[\code{r}'\mapsto 0] $,
  $\sigma_5= \sigma_4[\code{r}'\mapsto 1] $, and
  $\sigma_6=
  \sigma_5[\code{res}_0\mapsto1]$.
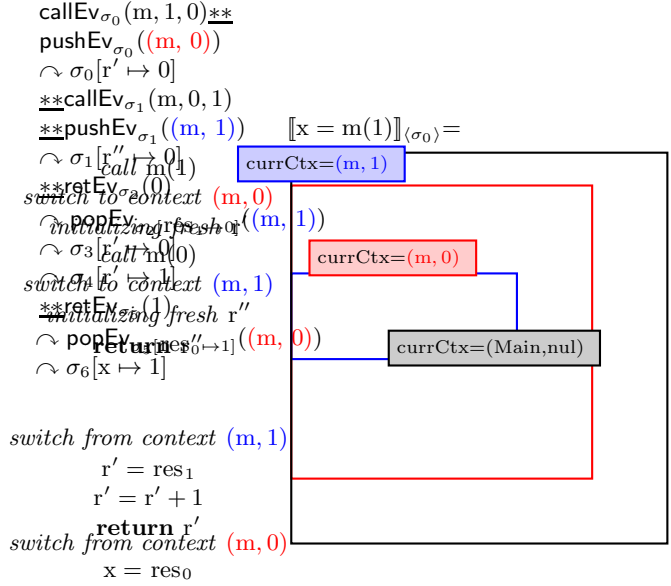
\begin{figure}
  \newcommand{\sOne}{\sigma_0[\code{r}'\mapsto 0]}
  \newcommand{\sTwo}{\sigma_1[\code{r}''\mapsto 0]}
  \newcommand{\sThree}{\sigma_2[\code{res}_1\mapsto 0]}
  \newcommand{\sFour}{\sigma_3[\code{r}'\mapsto 0]}
  \newcommand{\sFive}{\sigma_4[\code{r}'\mapsto 1]}
  \newcommand{\sSix}{\sigma_5[\code{res}_0\mapsto1]}
  \newcommand{\sSeven}{\sigma_6[\code{x}\mapsto 1]}

  \newcommand{\sOneShort}{\sigma_1}
  \newcommand{\sTwoShort}{\sigma_2}
  \newcommand{\sThreeShort}{\sigma_3}
  \newcommand{\sFourShort}{\sigma_4}
  \newcommand{\sFiveShort}{\sigma_5}
  \newcommand{\sSixShort}{\sigma_6}
  \newcommand{\sSevenShort}{\sigma_6}
  \newcommand{\ctxMZero}{\Red{(\code{m},\,0)}}
  \newcommand{\ctxMOne}{\Blue{(\code{m},\,1)}}
\centering  
\begin{tikzpicture}[remember picture]
  \node[align=left,text width=7cm] (beginM0) at (0,0){
   \vspace{4.6cm}\\
   \tikzmark{popCtxMZero}$\cons\popEvP{\sSix}{\ctxMZero}$\\
   \tikzmark{assignX}$\cons\sSeven$
  };
  \node[align=left, above of=beginM0,text width= 6.5cm, shift={(-0.2cm,-0.5 cm)}] (contentM0){
    \vspace{3.2cm}\\
    \tikzmark{popCtxMOne}$\cons\popEvP{\sThree}{\ctxMOne}$\\
    \tikzmark{assignRPrime}$\cons\sFour$\\
    \tikzmark{incrRPrime}$\cons\sFive$\\
    \tikzmark{returnRPrime}$\chopSem\retEvP{\sFiveShort}{1}$
  };
  \node[align=left, below of=contentM0, text width=5.5cm, shift={(-0.5cm,1.5 cm)}] (contentM1){
    \tikzmark{callMWithOne}$\callEvP{\sigma_0}{\code{m},1,0} \chopSem $\\
    \tikzmark{pushCtxMZero}$\pushEvP{\sigma_0}{\ctxMZero}$\\
    \tikzmark{initRPRime}$\cons\sOne$\\
    \tikzmark{callMWithZero}$\chopSem \callEvP{\sigma_1}{\code{m},0,1} $\\
    \tikzmark{pushCtxMOne}$\chopSem\pushEvP{\sigma_1}{\ctxMOne}$\\
  \tikzmark{initRSecond}$\cons \sTwo$\\
  \tikzmark{returnRSecond}$\chopSem\retEvP{\sTwoShort}{0}$
  };

  \node (callMWithOneNode) at (pic cs:callMWithOne) {};
  \node [above of=callMWithOneNode,shift={(1cm,-0.4cm)},align=left] {$\eval{\assignStmt{\code{x}}{\code{m}(1)}}{\singleton{\sigma_0}}$=};
  \node [left of=callMWithOneNode, shift={(-1cm,0.08cm)},align=left] {\emph{call} \code{m}(1)};
  \node (pushCtxMZeroNode) at (pic cs:pushCtxMZero) {};
  \node [left of=pushCtxMZeroNode, shift={(-1cm,0.08cm)},align=left] {\emph{switch to context} \ctxMZero};
  \node (popCtxMZeroNode) at (pic cs:popCtxMZero) {};
  \node [left of=popCtxMZeroNode, shift={(-1cm,0.08cm)},align=left] {\emph{switch from context} \ctxMZero};

  \node (callMWithZeroNode) at (pic cs:callMWithZero) {};
  \node [left of=callMWithZeroNode, shift={(-1cm,0.08cm)},align=left] {\emph{call} \code{m}(0)};
  \node (pushCtxMOneNode) at (pic cs:pushCtxMOne) {};
  \node [left of=pushCtxMOneNode, shift={(-1cm,0.08cm)},align=left] {\emph{switch to context} \ctxMOne};
  \node (popCtxMOneNode) at (pic cs:popCtxMOne) {};
  \node [left of=popCtxMOneNode, shift={(-1cm,0.08cm)},align=left] {\emph{switch from context} \ctxMOne};

  \node (assignRPrimeNode) at (pic cs:assignRPrime) {};
  \node [left of=assignRPrimeNode, shift={(-1cm,0.08cm)},align=left] {$\code{r}'=\code{res}_1$};
  \node (incrRPrimeNode) at (pic cs:incrRPrime) {};
  \node [left of=incrRPrimeNode, shift={(-1cm,0.08cm)},align=left] {$\code{r}'=\code{r}'+1$};
  \node (returnRPrimeNode) at (pic cs:returnRPrime) {};
  \node [left of=returnRPrimeNode, shift={(-1cm,0.08cm)}] {\Simple{return} $\code{r}'$};

  \node (assignXNode) at (pic cs:assignX) {};
  \node [left of=assignXNode, shift={(-1cm,0.08cm)},align=left] {$\code{x}=\code{res}_0$};

  \node (initRSecondNode) at (pic cs:initRSecond) {};
  \node [left of=initRSecondNode, shift={(-1cm,0.08cm)}] {\emph{initializing fresh} $\code{r}''$};
  \node (returnRSecondNode) at (pic cs:returnRSecond) {};
  \node [left of=returnRSecondNode, shift={(-1cm,0.08cm)}] {\Simple{return} $\code{r}''$};
  \node (initRPRimeNode) at (pic cs:initRPRime) {};
  \node [left of=initRPRimeNode, shift={(-1cm,0.08cm)}] {\emph{initializing fresh} $\code{r}'$};

  \node[draw=red,thick, minimum width = 4cm, minimum height= 3.90cm,shift={(1.9cm,-1.7cm)}] at (pushCtxMZeroNode) {};
  \node[draw=blue,thick,minimum width = 3cm, minimum height= 1.14cm,shift={(1.4cm,-0.33cm)}] at (pushCtxMOneNode) {};
  \node[draw=black,thick,minimum width = 5cm, minimum height= 5.2cm,shift={(2.4cm,-2.3cm)}] at (callMWithOneNode) {};

  \node[align=center,draw=blue, thick,fill=blue!20,below of=contentM1,shift={(1 cm,0.15 cm)}] (labelCtx1){
    {\scriptsize currCtx=\ctxMOne}
  };

  \node[align=center,thick,draw=red,fill=red!20,below of=contentM0,shift={(1.45 cm,-0.6 cm)}] (labelCtx0){
    {\scriptsize currCtx=\ctxMZero}
  };

  \node[align=center,thick,draw=black,fill=black!20,below of=contentM0,shift={(2.8 cm,-1.8 cm)}] (labelCtx0){
    {\scriptsize currCtx=(Main,nul)}
  };
\end{tikzpicture}
\caption{Visualization of the semantics of $\assignStmt{\code{x}}{\code{m}(1)}$ in $\sigma$\textsubscript{$0$}}
\label{fig:big-step-semantics-example}
\end{figure}
\end{example}
We only consider traces that are \emph{adequate}, i.e. consistent with local evaluation and
composition rules:
context switches can only occur in event pairs \callEv-\pushEv and
\retEv-\popEv consistently with the current context, and the freshness
of call identifiers has to be  preserved. We define $\Traces$ as the set of all adequate traces.
  For details, see
Appendix~\ref{app:soundness-calculus}.




\section{A Logic for Trace Contracts}
\label{sec:logic-trace-contracts}

We present a logic for specifying properties over finite program
traces.
The logic is a temporal $\mu$-calculus \cite{sti-93-handbook-chapter}
with two binary temporal operators, corresponding to concatenation and
chop over sets of traces, respectively.
We consider the syntax fragment without negation, which guarantees
that fixed-point formulas indeed denote fixed points of the
corresponding sematic transformers, and with least fixed-point
recursion only. Then we show this logic to be suitable for expressing
relevant finite-trace properties of recursive programs.



\subsection{Syntax}

The formulas of our logic are built from a set~$\LVar$ of first-order 
(``logical'') variables and a set~$\RVar$ of recursion variables.
The syntax of the logic is defined by the following grammar:
$$ \Phi \>::=\> \stateFml{P} \mid X (\overline{t}) \mid \mathit{Ev} \mid
                \Phi\wedge\Phi \mid \Phi\vee\Phi \mid 
                \Phi\concat\Phi \mid \Phi\chop\,\Phi \mid 
                (\muFml{X (\overline{y})}{\Phi}) (\overline{t})
$$
where~$P$ ranges over first-order predicates, $X$ over recursion 
variables, $\overline{y}$ over tuples of first-order variables, 
$\overline{t}$ over tuples of terms over first-order variables, 
and where in the last clause the arities of $\overline{y}$ and 
$\overline{t}$ agree.

Events $\mathit{Ev}$ have the form $\startEv(m,e,i)$ or
$\finishEv(m,e,i)$.  There are no push or pop events, but the finish
events record the context. As stated in the introduction, we aim at
procedure-modular specification and verification. Therefore, it is not
necessary to record all context switches in a global trace, as long as
corresponding calls and returns can be uniquely identified.

\begin{remark}
  Events $Ev$ in the logic are \emph{syntactic} representations of
  events in the local semantics
  (Sect.~\ref{sec:local-trace-semantics}), thus different
  entities. This permits to tailor events to the desired degree of
  abstraction of specifications and the deduction system. For example,
  if one is only interested in the interface behavior of a program,
  but not in internal state changes, one might choose abstract events
  of the form $\callEv(m)$/$\retEv(m)$, where call identifiers and
  arguments are dropped.
\end{remark}

\begin{example}
  To illustrate our logic, we introduce a formula template (or
  pattern) that we make extensive use of later.
  %
  Let $m$ be a procedure name, and let $\mathit{NoEv} (m)$ be a
  predicate that is true of a state if the latter is not an event that
  involves~$m$.
Consider the recursive formula:
$$ \Psi_m \>=\> 
   \muFml{X}{(\mathit{NoEv} (m) \vee \mathit{NoEv} (m) \concat X)} $$
   It is true for any trace that is either a singleton trace which is
   not an event involving~$m$ (base case), or else for a trace that
   starts with a singleton trace that is not an event involving~$m$
   and continues as a trace satisfying~$\Psi_m$ (induction case).
   Equivalently, $\Psi_m$~holds for finite traces not containing any
   event involving~$m$.
%
\end{example}

With help of~$\Psi_m$, we define the binary logic operator
$\Phi_1 \noEvent{m} \Phi_2$ as shorthand for
$\Phi_1 \chop\, \Psi_m \chop\, \Phi_2$, expressing the trace property of
satisfying~$\Phi_1$ initially, satisfying~$\Phi_2$ in the end, and not
containing any event involving~$m$ in between.
This is the logic equivalent of the semantic operator 
$\finiteNoEv{\many{ev}}$ on traces introduced earlier.

\subsection{Semantics}

To define the semantics of our logic, we need
a first-order variable assignment $\beta:\LVar\rightarrow\domain$ and
a recursion variable assignment
$\rho:\RVar\rightarrow(\overline{\domain}\rightarrow 2^{\Traces})$
that assigns each recursion variable to a set of traces.
The (finite-trace) semantics~$\evalPhiS{\Phi}$ of formulas 
as a set of traces is inductively defined in Fig.~\ref{fig:semantics-formulas}.
\begin{figure}[t]
\begin{align*}
\evalPhiS{\stateFml{P}} &= \{\singleton{\sigma} \>|\> \sigma\in\States \wedge \sigma \models P \} \qquad\qquad
\evalPhiS{X (\overline{t})} = \rho (X) (\beta (\overline{t})) \\
\evalPhiS{\startEv(m,e,i)} &= \{\callEvP{\sigma}{m,e,i} \chopSem \pushEvP{\sigma}{(m,i)} \>|\> \sigma\in\States\} \\
\evalPhiS{\finishEv(m,e,i)} &= \{\retEvP{\sigma}{\valP{}{e}} \concat \popEvP{\sigma [\code{res}_{i} \mapsto \valP{}{e}]}{(m,i)} \>|\> \sigma\in\States\} \\
\evalPhiS{\Phi_1\wedge\Phi_2} &= \evalPhiS{\Phi_1}\cap\evalPhiS{\Phi_2} \qquad\quad
\evalPhiS{\Phi_1\vee\Phi_2} = \evalPhiS{\Phi_1}\cup\evalPhiS{\Phi_2} \\
\evalPhiS{\Phi_1\concat\Phi_2} &= \{\tau_1\concat\tau_2 \>|\> \tau_1\in\evalPhiS{\Phi_1}  \wedge \tau_2\in\evalPhiS{\Phi_2}\} \\
\evalPhiS{\Phi_1\chop\,\Phi_2} &= \{\tau_1\,\chopT\,\tau_2 \>|\> \tau_1\in\evalPhiS{\Phi_1} \wedge \tau_2\in\evalPhiS{\Phi_2}\} \\
\evalPhiS{(\muFml{X (\overline{y})}{\Phi}) (\overline{t})} &= 
\evalPhiS{\muFml{X (\overline{y})}{\Phi}} (\beta (\overline{t})) \\
\evalPhiS{\muFml{X (\overline{y})}{\Phi}} &= \sqcap\, \{ F : \overline{\domain}\rightarrow 2^{\Traces} \>|\> \lambda \overline{d}. \evalPhi{\Phi}{}{\beta [\overline{y} \mapsto \overline{d}]}{\rho [X \mapsto F]} \sqsubseteq F \}\\
\end{align*}
\begin{center}\vspace*{-2.25em}
$\sqsubseteq$ and $\sqcap$ denote point-wise set inclusion and intersection, respectively
\end{center}
\caption{Semantics of formulas}
\label{fig:semantics-formulas}
\end{figure}
%
%

By a \emph{trace formula} we mean a formula of our logic that is
closed with respect to both first-order and recursion variables.
Since the semantics of a trace formula does not depend on any variable
assignments, we sometimes use $\evalPhiNoArg{\Phi}$ to denote
$\evalPhiS{\Phi}$ for arbitrary $\beta$ and~$\rho$.
We also permit a slight extension, where we allow trace formulas to be
syntactically closed with respect to first-order operators, as long as
fixed-point formulas occur only in positive positions.

\mycomment{
\begin{definition}[Semantics of Trace Formulas]
  \label{def:semantics-trace-fml}
  Given a first-order variable assignment $\beta$,
  the semantics of a trace formula~$\Phi$ is the set of all traces
  $\evalPhiS{\Phi}$ for any first-order variable assignment~$\beta$
  and recursion variable assignment $\rho$, formally:
  $\evalPhiNoArg{\Phi} = \bigcup_{\beta,\rho}\evalPhiS{\Phi}$.
\end{definition}
} 


\subsection{Specifying Procedure Contracts}

In general, a (finite-trace) procedure contract should capture the
sequences of states and events that are allowed to occur as a result
of a call to that procedure. Since the procedure may recursively
call other procedures, such contracts need to be stated recursively. 
The base case(s) of a recursive contract should specify the traces
that do not involve any procedure calls. 
The induction case(s), on the other hand, should specify the remaining
traces, and use recursion variables, properly applied to arguments,
at the points where calls to procedures are made. 


While it can be cumbersome to propose a general pattern for procedure
contracts, we illustrate the idea on a particular class of contracts
that generalize traditional, Hoare-style state-based contracts to
trace-based ones. As usual, we use state predicates over logical
variables to formulate pre- and postconditions that express the
intended relationship between the values of the program variables upon
procedure call and return, but in addition, we capture the structure
and arguments of recursive procedure calls.
We propose the following pattern, where~$m$ is the name of the
specified procedure, $n$~the value of its (sole) formal parameter with
which it is called, and $i$ the call identifier:
%
\begin{align*}
  \mathsf{H}_m = \mu &X_m(n,i).\\
  & \bigl( 
\stateFml{\mathit{pre}_m^{\mathit{base}}}\chop\,\startEv(m,n,i)\finiteNoM{m}\finishEv(m,f_m(n),i)\chop\stateFml{\code{res}_i\doteq f_m(n)}\, \lor\\
                   &  \phantom{\bigl(}\stateFml{\mathit{pre}_m^{\mathit{step}}}\chop\,\startEv(m,n,i)\finiteNoM{m}X_m(\mathit{step}^{-1}(n),\#(i))\\
  & \hspace{42.5mm}\finiteNoM{m}\finishEv(m,f_m(n),i)\chop\stateFml{\code{res}_i\doteq f_m(n)}
\bigr)
\end{align*}

Both base and inductive case use state predicates
$\mathit{pre}_m^{\mathit{base}/\mathit{step}}$ to establish the
precondition. Both cases use the state predicate
$\code{res}_i\doteq f_m(n)$ to specify the return
value~$\code{res}_i$ of $m$ as a function~$f_m$ depending on the
value~$n$ of the formal parameter.
The inductive case also specifies a recursive call to~$m$. The first
argument is the inverse of the recursive step function, for example,
$\mathit{step}^{-1}=n-1$ when $\mathit{step}=n+1$. The symbol \#
ensures that the call context identifier is fresh.

\begin{example} 
  \label{ex:H-m}
  We illustrate the use of pattern~$\mathsf{H}_{\text{\scriptsize{\code{m}}}}$ to provide a
  contract for the procedure~$\code{m}$ from
  Example~\ref{lst:running-example}.
To specify that the returned value equals the value of the parameter,
we should define $\mathit{pre}_{\text{\scriptsize{\code{m}}}}^{\mathit{base}} = \code{n}\doteq 0$,
$f_{\text{\scriptsize{\code{m}}}}$ as the identity function,
$\mathit{pre}_{\text{\scriptsize{\code{m}}}}^{\mathit{step}} = \code{n}>0$, and 
$\mathit{step} (\code{n}) = \code{n} + 1$:
%
%
\begin{align*}
  \mu &X_{\text{\scriptsize{\code{m}}}}(\code{n},i). \bigl(\stateFml{\code{n} \doteq 0}\chop\startEv(\code{m},0,i)\finiteNoM{\text{\scriptsize{\code{m}}}}\finishEv(\code{m},0,i)\chop\stateFml{\code{res}_i\doteq 0}\,\lor\\
               & \stateFml{\code{n}>0}\chop\startEv(\code{m},\code{n},i)\finiteNoM{\text{\scriptsize{\code{m}}}}X_{\text{\scriptsize{\code{m}}}}(\code{n}-1,\#(i))\finiteNoM{\text{\scriptsize{\code{m}}}}\finishEv(\code{m},\code{n},i)\chop\stateFml{\code{res}_i\doteq \code{n}}\bigr)
\end{align*}
\end{example}

The traditional big-step semantics of state-based contracts can be
expressed simply as
$\mathsf{H}_m\subseteq\stateFml{\mathit{pre}_m^{\mathit{base}}\lor\mathit{pre}_m^{\mathit{step}}}\finiteNoM{}\stateFml{\code{res}_i\doteq
  f_m(n)}$.

In addition to the final state, we can also specify behavior related
to intermediate states.  The contract in Example~\ref{ex:H-m} can be
extended to relate the result of the internal recursive call to the
final state by replacing the trace after the recursive call
$X_{\text{\scriptsize{\code{m}}}}(\code{n}-1,\#(i))$ with
\[
  \chop\stateFml{\code{res}_{\#(i)}=\code{n}-1}\finiteNoM{\text{\scriptsize{\code{m}}}}\finishEv(\code{m},\code{n},i)\chop\stateFml{\code{res}_i\doteq
  \code{res}_{\#(i)}+1}
  \]

Trace formulas (by design) completely abstract away from a specified
program, to which they are only connected via events. The set
$\mathit{Ev}$ of events can be extended if needed. It is also
conceivable to expose parts of the program state in traces: Let
predicate $\expose{l,e,i}$ be true in all states $\sigma$, such that
local variable $l$ in the call identified by $i$ has value
$\beta(e)$. Then one can, for example, insert
``$\chop\stateFml{\expose{\code{r},\code{n}-1,i}}$'' right after the recursive call
to $X_{\text{\scriptsize{\code{m}}}}$ above to express that $\code{r}$ has the value $\code{n}-1$ immediately
after the call returns.



\section{A Calculus for Deductive Verification}
\label{sec:calculus}

To verify whether a program fulfills its trace contract, we design a
symbolic execution calculus that reduces programs to a sequence of
syntactic state updates containing recursive calls. These are matched
against a trace contract formula with a novel abstraction rule.

\subsection{Updates}

Updates \cite{ABBHS16} can be seen as explicit substitutions recording
state changes. An \emph{elementary update} assigns an expression $e$
to variable $v$, denoted by $\upP{v:=e}$.  We also admit \emph{event
  updates} $\upP{\mathsf{Ev}(\many{e})}$ to record that event
$\mathsf{Ev}$ with parameters $\many{e}$ occurred (here, $\mathsf{Ev}$
is $\startEv$ or $\finishEv$, but this is generalizable).  A composite
update $\update$ is a (possibly empty) sequence of elementary/event
updates, see the grammar on top of Fig.~\ref{fig:stmt-update-eval}.
%
Updates~$\update$ precede statements $s$ with the meaning that in
$\update s$ statement $s$ is evaluated under the state changes
embodied~by~$\update$.
\begin{example}
  \label{ex:statement-leading-updates}
  An expression like
  ``$\,\upP{\startEv(\code{m}',0,i)}\upP{\code{r}:=0}~\returnStmt{\code{r}}\,$''
  results after partial symbolic execution of a procedure $\code{m}'$,
  where only the return statement is left to be executed in a state
  where \code{r} is assigned value $0$.
\end{example}

\paragraph{Local Evaluation.}

We extend the local semantic evaluation rules to programs with leading
updates. Fig.~\ref{fig:stmt-update-eval} contains semantic rules for
expressions of the form $u\update s$, where $u$ is either an
elementary update or an event update.  If $\update$ is empty then
after the evaluation of $u$ the rules from
Sect.~\ref{sec:local-trace-semantics} apply. The rules are similar to
those for statements. The first rule is similar to the rule for
sequential composition.
Three rules evaluate elementary updates:
\begin{enumerate*}[label=(\alph*)]
\item the case corresponding to assignments;
\item if a result variable $\code{res}_i$ is assigned, then the update
  is simply ignored: it is redundant, because its evaluation always
  follows $\finishEv(m,e,i)$;
\item the semantics of procedure calls invokes the program semantics.
\end{enumerate*}
The last two rules evaluate event updates occurring in the deduction
rules.
They represent the start and end of the execution of $m$,
respectively, and generate suitable events.
%
%
\begin{figure}[t]
  \begin{align*}
    & \qquad\qquad\qquad\qquad \update  ::=\epsilon~|~\upl v:= e\upr\update~|~\upP{\mathsf{Ev}(\overline{e})}\update \quad (\epsilon~\text{empty sequence of updates})\\[0.6em]
  & \valP{}{u\update s}
     =\concatTr{\tau}{\contF{}{\update s}}, \ \text{if} \ \valP{}{u}=\concatTr{\tau}{\contF{}{\zero}}\\
  & \valP{}{\upP{v:=e}} = \valP{}{v=e}, \text{ with $v\neq \code{res}_i$}\qquad\qquad
   \valP{}{\upP{\code{res}_i:=e}} = \singleton{\sigma} \concat \cont{\zero}\\
  & \valP{}{\upP{v:=m(e)}} = \eval{v=m(e)}{\singleton{\sigma}}\concat \cont{\zero},~\text{when}~ \eval{v=m(e)}{\singleton{\sigma}}~\text{is defined}\\
  & \valP{}{\upP{\startEv(m,e,i)}} = \callEvP{\sigma}{m,e,i} \chopSem \pushEvP{\sigma}{(m,i)} \concat \cont{\zero}\\ 
  & \valP{}{\upP{\finishEv(m,e,i)}} =  \retEvP{\sigma}{\valP{}{e}} \concat \popEvP{\sigma [\code{res}_{i} \mapsto \valP{}{e}]}{(m,i)} \concat \cont{\zero}
\end{align*}
\caption{Update syntax and local evaluation of statements with leading updates}
\label{fig:stmt-update-eval}
\end{figure}

\begin{example}
  One step of local evaluation of Example~\ref{ex:statement-leading-updates} in a
  state $\sigma$ yields the expression
  ``$\callEvP{\sigma}{\code{m}',0,i} \chopSem
  \pushEvP{\sigma}{(\code{m}',i)} \concat
  \cont{\upP{\code{r}:=0}\returnStmt{\code{r}}}$''.
\end{example}

\paragraph{Updates over Expressions.}

Expressions $e$ are evaluated in the current program state,
represented by a preceding composite update~$\update$. To evaluate~$e$
under~$\update$, written $\update (e)$, we apply updates from inner-
to outermost. Applying an elementary update $\upP{v:=e'}$ to an
expression~$e$ corresponds to syntactic substitution. Events have no
effect on the value of expressions. We obtain the following rules:
\begin{align*}
  \update'u(e) = \update'(u(e)) \qquad  
  \upP{v:=e'}(e) = e[v/e']) \qquad 
  \upP{Ev(\_)}(e) = \epsilon(e) =  e
\end{align*}

\noindent The trace composition is performed by applying rule
(\ref{eq:ruleprogress-rule})
($\textsc{\small{P}\scriptsize{ROGRESS}}$), overloaded to support the
occurrence of updates. Therefore, the semantics is defined in a
similar manner as in Def.~\ref{def:glob-sem}:
\begin{definition}[Semantics of Programs with Updates] 
  We define the \emph{semantics} of a terminating program (undefined
  else) with leading updates as
  $$\eval{\update \,s}{\tau}=\tau', \ \text{if} \ \tau, \cont{\update \,s} \overset{*}{\to} \tau \chopSem \tau', \cont{\zero}, \text{ with } s \text{ possibly empty.}$$
\end{definition}

When symbolically executing the return statement we need to generate a
$\finishEv$ whose context matches the current context of the leading
update.  We retrieve the call context from the leading update with a
helper function:
\begin{definition}[Current Context for Updates]
  \[currCtx(\update) = 
  \begin{cases}
    (m,i)& \update=\update'\upP{\startEv(m,\_,i)}\\
    currCtx(\update')& \update=\update' \upl v:=e \upr\\
    currCtx(\update')& \update=\update' \upP{\startEv(m,\_,i)} \update''  \upP{\finishEv(m,\_,i)}\\
    (\Main,nul)& otherwise
  \end{cases}\]
\end{definition}

\begin{example}
  The context in which the return statement of
  Example~\ref{ex:statement-leading-updates} executes is:
  $currCtx$(\upP{$\startEv(\code{m}',0,i)$}\upP{\code{r}:=0}) =
  $currCtx$(\upP{$\startEv(\code{m}',0,i)$}) = $(\code{m}',i)$.
\end{example}

\subsection{Calculus for Straight-line Programs}

First we present a calculus for programs that do neither contain loops
nor procedure calls.  Our deductive proof system is a Gentzen-style
sequent calculus, where partial symbolic execution of a program is
represented by $\update s$, with $\update$ the executed part and $s$
the remaining part.  If we \emph{judge} $\update s$ to conform to a
trace specification $\Phi$ we write $\judge{\update\,s}{\Phi}$, where
the \emph{judgment} $\judge{\update\,s}{\Phi}$ is evaluated in a
current state $\sigma$, formally:
\begin{definition}[Judgment, Assertion, Sequent]\label{def:judgment}
  A \emph{judgment} has the shape $\judge{\update\,s}{\Phi}$, where
  $\update$ is an update, $s$ a program statement, and $\Phi$ a trace
  formula.
  An \emph{assertion} is either a closed first-order predicate~$P$ or a
  judgment.
  A \emph{sequent} has the shape
  $\sequent{}{\judge{\update\, s}{\Phi}}$, where $\Gamma$ is a set of assertions.
\end{definition}

\begin{definition}[Semantics of Sequents]
  \label{def:sequent-semantics}
  Let $\sigma$ be a state.  A first-order predicate~$P$ is true
  in~$\sigma$ if $\sigma \models P$, as usual in first-order logic.  A
  judgment $\judge{\update\, s}{\Phi}$ is true in $\sigma$, denoted
  $\sigma \models \judge{\update\, s}{\Phi}$, when 
  $\eval{\update\, s}{\singleton{\sigma}}$ is undefined or $\eval{\update\, s}{\singleton{\sigma}} \in \evalPhiNoArg{\Phi}$.  A
  sequent $\sequent{}{\judge{\update\, s}{\Phi}}$ is true in~$\sigma$
  if one of the assertions in $\Gamma$ is not true in $\sigma$ or
  $\sigma \models \judge{\update\, s}{\Phi}$.  A sequent is
  \emph{valid} if it is true in all states~$\sigma$.
\end{definition}

\mycomment{
\begin{definition}[Semantics of Sequents]
\label{def:sequent-semantics}
  Let $\sigma$ be a state and $\beta$ a first-order variable assignment. 
  A first-order predicate $P$ is true in $\sigma$ and $\beta$ if
  $(\sigma,\beta) \models P$ as usual in first-order
  logic. An assertion $\judge{\update\, s}{\Phi}$ is true in $\sigma$ and $\beta$,
  denoted $(\sigma,\beta) \models \judge{\update\, s}{\Phi}$, whenever
  $\evalB{\update\, s}{\sigma} \in \evalPhiB{\Phi}$. A sequent
  $\sequent{}{\judge{\update\, s}{\Phi}}$ is true in $\sigma$ and $\beta$ if one
  of the assertion in $\Gamma$ is not true in $\sigma$ and $\beta$ or 
  $(\sigma,\beta) \models \judge{\update\, s}{\Phi}$. A sequent is \emph{valid}
  if it is true in all states~$\sigma$ and first-order variable assignments~$\beta$.
\end{definition}
} 


\paragraph{Rules.}

The rules of our calculus for the symbolic execution of straight-line
programs are shown in Fig.~\ref{fig:straight-line-programs}.
\begin{figure}
  \centering
$
\begin{array}{@{}c@{\qquad}c@{}}
\seqRule{Assign}{
\sequent{}{\update \upl v:= e\upr\judge{s}{\Phi}}
}{
\sequent{}{\update \, \judge{v=e;s}{\Phi}}
}
& 
\seqRule{Prestate}{
\sequent{P}{Q}\qquad
\sequent{P}{\update \judge{s}{\Phi}}
}{
\sequent{P}{\update \judge{s}{\stateFml{Q}\,\chop\,\Phi}}
}
\\[0.35cm]

\seqRule{Scope}{
\sequent{}{\update \judge{s}{\Phi}}
}{
\sequent{}{\update  \judge{\{s\}}{\Phi}}
}
&
\seqRule{Cond}{
\sequent{\update (e)}{\judge{\update s; s'}{\Phi}} \qquad
\sequent{\update (!e)}{\judge{\update s'}{\Phi}}
}{
\sequent{}{\update \ \judge{\ifStmt{e}{s}; s'}{\Phi}}
}
\\[0.4cm]
\multicolumn{2}{c}{
\seqRule{Return}{
\mathrm{currCtx}(\mathcal{U})=(m,i)\qquad
\sequent{}{\mathcal{U} \upl\finishEv(m,e,i)\upr\judge{\code{res}_i=e}{\Phi}}
}{
\sequent{}{\mathcal{U}\,\judge{\mbox{\lstinline{return}}~e}{\Phi}}
}
}\\[0.45cm]
\multicolumn{2}{c}{
\seqRule{VarDecl}{
\sequent{}{\update \upl v':=0\upr\judge{\{s[v'/v]\}}{\Phi}}
\qquad
v'\text{ fresh for }s,\Gamma,\Phi
}{
\sequent{}{\update \judge{\{v;s\}}{\Phi}}
}
}
\end{array}
$
\caption{Sequent rules for straight-line programs}
\label{fig:straight-line-programs}
\end{figure}
%
%
%
\noindent
For space reasons, we omit the standard rules for Gentzen-style
first-order sequent calculi.
In the following we present a number of specific rules, starting with
the rule for unfolding fixed-point formulas:
$$
\seqRule{Unfold}{
\sequent{}{\update \judge{s}{\Phi \!\left[(\muFml{X (\overline{y})}{\Phi}) / X,\overline{t} / \overline{y}\right]}}
}{
\sequent{}{\update \judge{s}{(\muFml{X (\overline{y})}{\Phi}) (\overline{t})}}
}
$$

\begin{example}\label{ex:calculus:straight-line}
  We show the first step of the symbolic execution of the
  straight-line program ``\code{if (k!=0) \{ r=k-1; r=r+1; \} return
    r}'' with premise
  $\Gamma\equiv\code{k}>0$ and under a similar event update
  $\update\equiv\upP{\startEv(\code{m},\code{k},\code{i}')}$ as in
  Example~\ref{ex:statement-leading-updates}. The first statement is a
  conditional, so rule \ruleName{Cond} is applied. Observe that
  $\update(\code{k!=0})$ evaluates to $\code{k !=
    0}$ and is subsumed by $\code{k > 0}$.
  We show only the left premise as the right premise is immediately
  closed due to $\code{k > 0}$. The final sequent is the result of
  applying rule \ruleName{Assign} twice, followed by an application of
  the \ruleName{Return} rule. An unabbreviated version is in
  appendix~\ref{app:extended-examples}.
\begin{prooftree}
  \AxiomC{\LeftLabel{}\text{(symbolic execution finished)}} 
  \UnaryInfC{\RuleL{}$\seq{\code{k}>0}{\judge{\update\upP{\code{r}:=\code{k}-1}\upP{\code{r}:=\code{r}+1} \upP{\finishEv(\code{m},\code{r},\code{i}')}\upP{\code{res}_0:=\code{r}}}{\Phi}}$}  
  \def\extraVskip{-0.25pt}
  \UnaryInfC{$\vdots$} 
  \def\extraVskip{1.25pt}
  \UnaryInfC{\RuleL{Cond}$\seq{\code{k}>0}{\judge{\update\,\mbox{\lstinline|if (k!=0) \{r=k-1; r=r+1;\} return r|}}{\Phi}}$} 
\end{prooftree}
\end{example}


\mycomment{
We want to prove that for any $\sigma$, $\eval{\upP{v:=e};s}{\sigma}=\eval{v:=e;s}{\sigma}$.
By Proposition~\ref{prop:glob-sem-comp}, we have that 
$$\eval{r;s}{\sigma}=\eval{r}{\sigma} \chopSem \eval{s}{\sigma'}, \ \text{where} \ last(\eval{r}{\sigma})=\sigma'$$
By Proposition~\ref{prop:glob-sem-comp-update}, we have that 
$$\eval{\update s}{\sigma}=\eval{\update }{\sigma} \chopSem \eval{s}{\sigma'}, \ \text{where} \ last(\eval{\update }{\sigma})=\sigma'$$
Since $$\Blue{\eval{\upP{v:=e}}{\sigma} = \singleton{\sigma}\cons \sigma[v\mapsto\valB{\sigma}{e}] = \eval{v:=e}{\sigma}}$$ we have
\begin{align*}
  &\eval{\upP{v:=e}s}{\sigma}=\\
  &\Blue{\eval{\upP{v:=e}}{\sigma}} \chopSem \eval{s}{\sigma[v\mapsto\valB{\sigma}{e}]}=\\
  &\Blue{\singleton{\sigma}\cons \sigma[v\mapsto\valB{\sigma}{e}]} \chopSem \eval{s}{\sigma[v\mapsto\valB{\sigma}{e}]}=&\text{\emph{\footnotesize (this step can be omitted)}}\\
  &\Blue{\eval{v:=e}{\sigma}} \chopSem \eval{s}{\sigma[v\mapsto\valB{\sigma}{e}]}=\\
  &\eval{v:=e;s}{\sigma}
\end{align*}
}

\subsection{Procedure Contracts}
\label{sec:invar-spec-trac}

Specifying and verifying trace contracts for each procedure allows us
to verify procedure calls in a modular way.  First we show how to
specify a procedure $m$ with a contract $\mathbf{C}_m$, then we
present the rule $\ruleName{ProcedureContract}$ that is used to prove
$\mathbf{C}_m$ in our calculus.
Let
\begin{equation}
\mathbf{C}_m=\forall
n,i.(pre_m(n)\rightarrow\judge{m(n)}{\Phi_m(n,i)\chop \stateFml{\code{res}_i\doteq f_m(n)}})\enspace,\label{eq:contract}
\end{equation}
where $f_m(n)$ is the result of $m$ given input $n$.\footnote{For technical
reasons, the formalization of contract $\mathbf{C}_m$ uses a procedure
call outside of an assignment. Its local trace semantics
$\valP{}{m(e)} =
\concatTr{\callEvP{\sigma}{m,\valP{}{e},id}}{\contF{}{\zero}}$ is the
same as for procedure calls with assignment, except for the empty
continuation.}
%
Eq.~\eqref{eq:contract} can be seen as the generalization of
$\forall n.(\{pre_m(n)\}m(n)\{\code{res}\doteq f_m(n)\})$, a state-based
contract, where $n$ is a first-order parameter.
%
%
When proving correctness of trace contracts $\mathbf{C}_m$ of
recursive procedures $m$ \cite{ABBHS16,Hofmann97,Oheimb99}, one
establishes that $\Phi_m$ is a \emph{specification invariant}
for the implementation of~$m$. One proves that the inlined body of $m$
respects $\mathbf{C}_m$ and when symbolic execution arrives at a
recursive call to $m$, one can assume that $m$ holds already. This
yields partial correctness.  We generalize this approach to
traces in the following rule for recursive self-calls:
$$
\seqRule{ProcedureContract}{
\sequent{pre_m(n'),\,\mathbf{C}_m}{\judge{\inline(m,n',i')}{\Phi_m(n',i')}}
\qquad i',n'\text{ fresh}
}{
\sequent{}{\mathbf{C}_m}
}
$$

The rule expresses that for any parameter $n'$ and any recursion depth
$i'$, the trace specification $\Phi_m(n',i')$ is an invariant for the
inlined procedure body, where $\mathbf{C}_m$ can be assumed in recursive
calls. To represent inlining succinctly, we use
\begin{equation}
\label{eq:inline}
\inline(m,e,i) = \upl\startEv(m,e,i)\upr\,e'=e; mb[e'/p]\enspace,
\end{equation}
where $p$ is the procedure parameter, and $e'$ is fresh.

In general, there might be more than one procedure call, so the
assumption in (\textsf{ProcedureCall}) should actually be
$\bigwedge_{m\in\methods{P}}\mathbf{C}_m$.

\begin{example}\label{ex:procedure-contract-rule}
  Symbolic execution for procedure \code{m} from
  Example~\ref{lst:running-example} with contract
  $ \mathbf{C}_{\text{\scriptsize{\code{m}}}}=\forall
  n,i.(\judge{\code{m}(n)}{\Phi_{\text{\scriptsize{\code{m}}}}(n,i)\chop
    \stateFml{\code{res}_i\doteq n}})$, with
  $\update=\upP{\startEv(\code{m},\code{n}',\code{i}')}$:
  \begin{prooftree}
    \AxiomC{(here symbolic execution of procedure body starts)}
    \UnaryInfC{\RuleL{}$\seq{\code{n}'\geq 0,\,\mathbf{C}_{\text{\scriptsize{\code{m}}}}}{\judge{\update\upP{\code{k}':=\code{n}'}\upP{\code{r}':=0} \code{s}[\code{k}/\code{k}',\code{r}/\code{r}']}{\Phi_{\text{\scriptsize{\code{m}}}}(\code{n}',\code{i}')}}$}
	\noLine
	\def\extraVskip{-2.5pt}
	\UnaryInfC{$\vdots$} 
	\noLine
    \def\extraVskip{2pt}
    \UnaryInfC{$\seq{\code{n}'\geq 0,\,\mathbf{C}_{\text{\scriptsize{\code{m}}}}}{\judge{\inline(\code{m},\code{n}',\code{i}')}{\Phi_{\text{\scriptsize{\code{m}}}}(\code{n}',\code{i}')}}$}
    \UnaryInfC{\RuleL{ProcedureContract}$\seq{}{\mathbf{C}_{\text{\scriptsize{\code{m}}}}}$}
  \end{prooftree}
\end{example}

\subsection{Procedure Calls}
\label{sec:trace-abstraction}

As Example~\ref{ex:procedure-contract-rule} shows, after applying rule
\ruleName{(ProcedureContract)}, a procedure body can be fully symbolically
executed, where recursive calls in assignments are simply handled by
the \ruleName {(Assign)} rule. The semantics of elementary updates with a
call on the right ensures that this is sound provided that we make the
assumption that a procedure call has no side effects.
%
In this way, we can retrofit standard, state-based verification into
our trace-based framework.\footnote{We stress that the absence of side
  effects is for ease of presentation and not a fundamental limitation of our approach. How to model side effects in symbolic execution is well-known~\cite{ABBHS16}.}
%
Complete symbolic execution of a procedure body with a recursive call yields a judgment
of the form
\begin{equation}
  \judge{\mathcal{U}_1\upl r:=m(e)\upr\,\mathcal{U}_2}{\Phi_1\chop\,\Phi_m(e,k)\chop\,\Phi_2}\enspace.\label{eq:judge-shape}
\end{equation}
The updates are here followed by the ``empty'' program.
For the trace specification the above shape is also justified, because
typically we have
$\Phi_m=\mu X(y).(\Phi_1(y)\lor\cdots\lor\Phi_n(y))$. This
specification one strengthens into $\mu X(y).\Phi_i(y)$, where
$\Phi_i$ is the specification case corresponding to the current
symbolic execution path.  At this point, and after symbolic execution
finished, we can use pattern~\eqref{eq:judge-shape} and assumption
$\mathbf{C}_m$ in the following \emph{trace abstraction rule}:
$$
\seqRule{TrAbs}{
\sequent{}{\judge{\mathcal{U}_1}{\Phi_1}}\qquad
\sequent{}{\mathcal{U}_1(pre_m(e)})\qquad
\seq{\mathbf{C}_m}{\judge{\upP{v:=f_m(\mathcal{U}_1(e))}\mathcal{U}_2}{\Phi_2}}
}{
\sequent{\mathbf{C}_m}{\judge{\mathcal{U}_1\upl v:=m(e)\upr\,\mathcal{U}_2}{\Phi_1\chop\,\Phi_m(e,k)\chop\,\Phi_2}}
}
$$
where $f_m(\cdot)$ is the function computed by
$m$.  The trace abstraction rule decomposes the conclusion into three
premises. The first premise verifies that the trace represented by
$\update_1$ conforms to $\Phi_1$. The second premise guarantees that
the precondition of contract $\mathbf{C}_m$ is satisfied.  Contract
$\mathbf{C}_m$ is now implicitly used to guarantee that the trace of
update $\upP{v:=m(e)}$ conforms to $\Phi_m(e,k)$, up to assignment of
the procedure's result to program variable $v$. The task of the third
premise is then to ensure that the trace consisting of the latter
assignment followed by update $\update_2$ conforms to $\Phi_2$.

%
\begin{example}\label{ex:trace-abstraction}
  Trace abstraction is applied in the continuation of
  Example~\ref{ex:procedure-contract-rule}, after the recursive call
  in the body was symbolically executed and moved to an update. At
  this point the goal sequent has the form
  $$
  \code{n}'>0, \mathbf{C}_{\text{\scriptsize{\code{m}}}}\vdash \upP{\update_1}\upP{\code{r}':=\code{m}(\code{n}'-1)}\update_2 : \Phi_1\finiteNoM{\text{\scriptsize{\code{m}}}}\Phi_{\text{\scriptsize{\code{m}}}}(\code{n}'-1,\#(\code{i}'))\finiteNoM{\text{\scriptsize{\code{m}}}}\Phi_2
  $$
  with $\update_1\equiv\upP{\startEv(\code{m},\code{n}',\code{i}')}$,
  $\update_2\equiv\upP{\code{r}'\!:=\!\code{r}'+1}\upP{\finishEv(\code{m},\code{r}',\code{i}')}\upP{\code{res}_{\code{i}'}\!:=\!\code{r}'}$,
  $\Phi_1\equiv\stateFml{\code{n}'>0}\chop\startEv(\code{m},\code{n}',\code{i}')\finiteNoM{\text{\scriptsize{\code{m}}}}$,
  $\Phi_{\text{\scriptsize{\code{m}}}}(\code{n}'-1,\#(\code{i}'))\equiv X_{\text{\scriptsize{\code{m}}}}(\code{n}'-1,\#(\code{i}'))$, and
  $\Phi_2\equiv\finiteNoM{\text{\scriptsize{\code{m}}}}\finishEv(\code{m},\code{n}',\code{i}')\concat\stateFml{\code{res}_{\code{i}'}\doteq
    \code{n}'}$.
  Applying the trace abstraction rule results in the following three sequents:
  \begin{align*}
    \code{n}'>0 & \vdash \judge{\upP{\startEv(\code{m},\code{n}',\code{i}')}}{\stateFml{\code{n}'>0}\chop \startEv(\code{m},\code{n}',\code{i}')\finiteNoM{\text{\scriptsize{\code{m}}}}}\\
\code{n}'>0 & \vdash \code{n}'>0 \\
    \code{res}_{\code{i}'} \doteq \code{n}'-1 & \vdash \upP{\code{r}':=\code{res}_{\code{i}'}}\upP{\code{r}':=\code{r}'+1}\upP{\finishEv(\code{m},\code{r}',0)}\upP{\code{res}_{\code{i}'}:=\code{r}'}:\\[-0.5ex]
                &\hspace*{5cm}\finiteNoM{\text{\scriptsize{\code{m}}}}\finishEv(\code{m},\code{n}',\code{i}')\concat\stateFml{\code{res}_{\code{i}'}\doteq \code{n}'}
  \end{align*}
  The first subgoal is provable with rule (\textsf{Prestate}) and a
  simplification rule, the second is trivial, the third requires
  simplification rules found in Appendix~\ref{app:extended-examples}.
  %
\end{example}

\paragraph{Loops.}
For space reasons, we do not provide rules for dealing with loops.
Conceptually, it is well-known that contracts can be used to specify
loops.  A systematic overview and comparison between invariant-based
and contract-based loop specification is in \cite{Ernst22}. A suitable
adaptation of our contract rules to the case of loops will be the
topic of future work.


\subsection{Soundness}

\begin{definition}[Soundness]
A rule of the calculus is \emph{sound} if the validity 
of the conclusion follows from the validity of the premises.
A calculus is sound if it can  prove only valid statements.
\end{definition}


\begin{theorem}[Calculus soundness]\label{soundness:calculus}
The presented sequent calculus is sound.
\end{theorem}

\begin{proof}
Direct consequence of the local soundness of each rule. Proof sketches for the soundness of the rules are given in  Appendix~\ref{app:proofs}.
\end{proof}



\section{Related Work}
\label{sec:related-work}

We specify symbolic traces, so it is not surprising to find
related work in extensions of LTL model checking and program
synthesis. CaReT logic~\cite{AEM04} can specify the call structure of
programs which is modeled as pushdown automata.  It has
abstract versions of the temporal next and until operators that jump
over balanced calls. The call and event structure is fixed.
In~\cite{FrohmeSteffen21a} \emph{Systems of Procedural Automata} model procedure calls with context-free rules
for atomic actions and procedure calls. The
goal is to learn automata from observed traces.
Temporal Stream Logic~\cite{FKPS19} features uninterpreted function
terms and updates in addition to standard LTL operators, 
aiming at program synthesis of B\"uchi stream automata and 
FPGA programs.
In each case the setup is finite or propositional, not first-order.
Matching and fixed-point specifications are not possible.

Process logic~\cite{HarelKP80} and interval temporal logic
\cite{HalpernShoham91} feature the chop operator, which was taken up
by Nakata \& Uustalu~\cite{NakataUustalu15}, who used infinite
symbolic traces to characterize non-terminating loops. These were
extended to a rich dynamic logic~\cite{BubelCHN15} and equipped with
events and a local trace semantics \cite{DHJPT17}.

Cyclic proof systems to prove inductive claims, including contracts of
recursive procedures, date back to Hoare's axiomatization of recursive
procedures \cite{Hoare71}. Gurov \& Westman \cite{gur-wes-18}
provide an abstract framework based on denotational semantics to 
formally justify (cyclic) procedure-modular verification, while
Brotherston \& Simpson \cite{BrotherstonSimpson11} investigate the
expressive power of sequent calculi for cyclic and infinite arguments
that are often used as a basis in deductive verification. Recursive
predicate specifications are standard, for example, in separation
logic \cite{ParkinsonBierman05}, but not used to specify program
traces.  Several papers \cite{DamGurov02,SprengerDam03b} present a
first-order $\mu$-calculus, but do not feature explicit programs or
events.

Constructive logical frameworks \cite{Coq04,Isabelle02} feature
abstract specifications in terms of typed higher-order logic formulas,
but they target functional languages and feature neither events nor
traces.



\section{Conclusion and Future Work}
\label{sec:concl-future-work}

In this paper we established the fundamental theory of trace-based
contracts, generalizing specification and deductive verification with
state-based contracts. The ingredients are
\begin{enumerate*}[label=(\roman*)]
\item an expressive fixed-point logic to characterize
  complex event structures over recursive procedures;
\item a uniform, local trace-event semantics for programs, state
  updates, and the trace logic;
\item a sound symbolic execution calculus with rules to prove
  trace contracts and programs with recursive procedures.
\end{enumerate*}
Programs and trace contracts communicate semantically via a
configurable set of events. This permits fully abstract specification
of programs and a highly flexible approach to specify concepts like
user input, concurrency, etc.
Indeed, to harvest these opportunities arising from an expressive,
trace-based specification and verification approach, will be the topic
of follow-up papers, where we will look at concurrent programs, loops,
and complex case studies.

Further interesting questions concern the precise expressivity of our
trace logic, in particular, which events are required for completeness.





\bibliography{reiner}
\bibliographystyle{abbrv}


\clearpage
\appendix

\section{Proofs}
\label{app:proofs}

\subsection{Trace Adequacy}
\label{sect:trace-adequacy}

We only consider traces consistent with local evaluation and
composition rules: at most one variable update occurs between two
consecutive states and the events in a trace properly record call
events, context switches, and return events.

\begin{definition}[Trace Adequacy] 
  \label{def:trace-adequacy}
  We say $\tau = \tau_1 \chopSem \tau_2$ is an \emph{adequate} trace
  if either $\tau_1=\tau_2=\singleton{\sigma}$ or $\tau_1$ is adequate
  and exactly one of the following holds:
  \begin{enumerate}
  \item $\tau_2 = \consTr{\singleton{\sigma}}{\sigma[x\mapsto
      v]}$ \label{def:item:adequacy-update}
  \item $\tau_2 = \callEvP{\sigma}{\_,\_,id}$,
    $\tau_1\in\finiteNoEv{\pushEvP{}{(\_,id)},\,\popEvP{}{(\_,id)},\,\callEvP{}{\_,\_,id}}$
    and\\
    $\lastEv{\tau_1}\not\in\{\callEv,\retEv\}$\label{def:item:adequacy-callEv}
  \item $\tau_2 = \retEvP{\sigma}{\_}$,
    $\lastEv{\tau_1}\not\in\{\callEv,\retEv\}$\label{def:item:adequacy-retEv}
  \item $\tau_2 = \pushEvP{\sigma}{(m,id)}$,
    $\tau_1 \in\finiteNoEv{} \,
    \callEvP{\sigma}{m,\_,id}$\label{def:item:adequacy-pushEv}
  \item $\tau_2 = \popEvP{\sigma}{(m,id)}$,
    $\tau_1 \in\finiteNoEv{} \, \retEvP{\sigma}{m}$,
    $currCtx(\tau_1)= (m,id)$.\label{def:item:adequacy-popEv}
  \end{enumerate}
  $\Traces$ is the set of all adequate traces.
\end{definition}

The second clause expresses that call identifiers are unique as well
as one part of the requirement that call/return events are followed
immediately by push/pop events. That requirement is ensured together
with the remaining clauses.


The events $\callEv$ and $\retEv$ are always followed by $\pushEv$ and $\popEv$ respectively.
Therefore they can be the last event of a trace only if they occur at the end of the trace.
This result is formalized in the following lemma.
\begin{lemma}\label{lemma:lastEv}
  If $\semSeq{\semPair{\singleton{\sigma}}{s}}{\semPair{\tau}{s'}}{n}$ and $\callEv=lastEv(\tau)$, or $\retEv=lastEv(\tau)$,
  then $\tau=\finiteNoEv{} \, \callEv_{last(\tau)}$ or $\tau=\finiteNoEv{}\, \retEv_{last(\tau)}$, respectively.
\end{lemma}

If $\semSeq{\semPair{\tau}{s}}{\semPair{\tau'}{s'}}{*}$ in exactly $n$ steps then we write $\semSeq{\semPair{\tau}{s}}{\semPair{\tau'}{s'}}{n}$.

\begin{lemma}\label{lemma:n-adequacy}
  Given a program $s$, if $\semSeq{\semPair{\singleton{\sigma}}{s}}{\semPair{\tau}{s'}}{n}$ then $\tau$ is adequate.
\end{lemma}
\begin{proof}
  We prove the lemma by induction on $n$, i.e. on the number of applications of composition rules.

  \textbf{Base case (n=1)}.
  For non trivial programs, $s$ has the form ``$\code{x} \ d; r$'', i.e. at least a variable $\code{x}$ is declared. 
  Otherwise, assignment and procedures calls cannot occur. Therefore 
  if $n=1$ then the rule applied is the \emph{progress rule}.
  $$\semSeq{\semPair{\singleton{\sigma}}{\code{x}\, d; r}}{\semPair{\singleton{\sigma}\cons \sigma[\code{x}\mapsto 0]}{d; s}}{}$$
  where $\singleton{\sigma}\cons \sigma[\code{x}\mapsto 0]$ satisfies condition (\ref{def:item:adequacy-update}) of Def.~\ref{def:trace-adequacy}. 

  \textbf{Inductive Step.}
  Let's assume as IH that $\semSeq{\semPair{\singleton{\sigma}}{s}}{\semPair{\tau}{s'}}{n}$ with $\tau$ adequate and $s'\neq \zero$.
  Let also $\sigma'=last(\tau)$.
  There are three main cases: 
  \begin{itemize}
    \item $\tau=\tau'\chopSem\callEvP{\sigma'}{m,\_,i}$:  condition~(\ref{def:item:adequacy-pushEv}) is satisfied since applying \emph{call rule} we have
    $$\semSeq{\semPair{\tau}{s'}}{\semPair{\tau'\chopSem\callEvP{\sigma'}{m,\_,i}\chopSem \pushEvP{\sigma'}{(m,i)}}{s''}}{}$$

    \item $\tau=\tau'\chopSem\retEvP{\sigma'}{\_}$: condition~(\ref{def:item:adequacy-popEv}) is satisfied since applying \emph{return rule} we have
    $$\semSeq{\semPair{\tau}{s'}}{\semPair{\tau'\chopSem\retEvP{\sigma'}{\_ }\chopSem \popEvP{\sigma'}{(m,id)}}{s''}}{}$$
    where $currCtx(\tau') = (m,id)$ by definition.

    \item $\tau\neq\tau'\chopSem\callEv_{\sigma'}$ and $\tau\neq\tau'\chopSem\retEv_{\sigma'}$: the \emph{progress rule} applies and we have three cases
    \begin{itemize}
      \item $\semSeq{\semPair{\singleton{\sigma}}{s}}{\semPair{\tau}{x:=m(\many{e});r)}}{n}$:  a \callEv with a fresh call identifier is generated.
      Since,  by Lemma~\ref{lemma:lastEv} we have that $last(\tau) \not\in \{\callEv, \retEv\}$ condition~(\ref{def:item:adequacy-callEv}) is satisfied.
      \item $\semSeq{\semPair{\singleton{\sigma}}{s}}{\semPair{\tau}{\returnStmt{e};r)}}{n}$:  a \retEv is generated.
      Since,  by Lemma~\ref{lemma:lastEv} we have that $last(\tau) \not\in \{\callEv, \retEv\}$ condition~(\ref{def:item:adequacy-callEv}) is satisfied.
      \item Otherwise  we have $\semSeq{\semPair{\tau}{s'}}{\semPair{\tau'}{s'')}}{n}$ where either $\tau'=\tau$ adequate by IH, or
      $\tau'=\tau\cons\sigma'[x\mapsto v]$, that satisfies condition~(\ref{def:item:adequacy-update}).
    \end{itemize}
  \end{itemize}
  Therefore, given $$\semSeq{\semPair{\singleton{\sigma}}{s}}{\semPair{\tau}{s'}}{n}$$ with $\tau$ adequate 
  we have $$\semSeq{\semPair{\singleton{\sigma}}{s}}{\semPair{\tau'}{s''}}{n+1}$$ where $\tau'$ is adequate.
  \qed
\end{proof}

\begin{theorem}[Adequacy of Program
  Semantics]\label{theorem:adequacy-semantics} 
  The program semantics $\eval{s}{\singleton{\sigma}}$ is an adequate
  trace for any state $\sigma$ and terminating program $s$.
\end{theorem}
\begin{proof}
  From Lemma~\ref{lemma:n-adequacy} it follows that given a program $s$ if $\semSeq{\semPair{\singleton{\sigma}}{s}}{\semPair{\tau}{\zero}}{*}$ 
  then $\tau$ is adequate. In other words $\eval{\sigma}{s}$ is adequate.
  \qed
\end{proof}

\subsection{Soundness of the Calculus}
\label{app:soundness-calculus}

\begin{proposition} 
   \label{prop:glob-sem-comp}
     $\eval{r;s}{\tau}=\tau' \chopSem \eval{s}{\tau'}$, with $\tau' = \eval{r}{\tau}$.
 \end{proposition}
 
\begin{proposition}
  \label{prop:glob-sem-comp-update}
  $\eval{\update s}{\tau} = \eval{\update}{\tau} \,\chopSem\, \eval{s}{\eval{\update}{\tau}}$ and  $\eval{\update \update'}{\tau} = \eval{\update}{\tau} \,\chopSem\, \eval{\update'}{\eval{\update}{\tau}}$,
  and
  $\eval{\upP{v := m(e)}}{\tau} = \eval{\upP{v := m(e)}}{last(\tau)}$.
\end{proposition}

\begin{theorem}[Soundness]
  The presented sequent calculus is sound.
\end{theorem}

\begin{proof}
The result is a direct consequence of the (local) soundness of each rule, 
which we show here for selected rules. 
Recall that a rule is sound if its conclusion is a valid sequent
whenever all its premises are.
For some rules we also show \emph{reversibility} (also called
\emph{backward soundness}), which means that whenever the conclusion
of the rule is a valid sequent, the rule can be applied backwards in a
way so that all premises are valid.


\paragraph{Rule $\mathsf{Assign}$.}
We shall prove soundness and reversibility of the rule.
  Let $\sigma$ be a state. We have, with $\tau=\eval{\update}{\singleton{\sigma}}$:
  %
  $$ \begin{array}{lll}
     & 
  \sigma \models \judge{\update\upP{v:=e}s}{\Phi} & \\
  \Leftrightarrow~ & 
  \eval{\update\upP{v:=e}s}{\singleton{\sigma}} \in \evalPhiNoArg{\Phi} &
  \mbox{\{Def.~\ref{def:sequent-semantics}\}} \\
  \Leftrightarrow & 
  \tau \,\chopSem\, 
  \eval{\upP{v:=e}}{\tau} \,\chopSem\, 
  \eval{s}{\eval{\upP{v:=e}}{\tau}} 
  \in \evalPhiNoArg{\Phi} &
  \mbox{\{Prop.~\ref{prop:glob-sem-comp-update}\}} \\
  \Leftrightarrow & 
  \tau \,\chopSem\, 
  \eval{v=e}{\tau} \,\chopSem\, 
  \eval{s}{\eval{v=e}{\tau}} 
  \in \evalPhiNoArg{\Phi} ~~~~~~~~~~~~~&
  \mbox{\{Def. $\valP{}{\upP{v:=e}}$\}} \\
  \Leftrightarrow & 
  \eval{\update v=e;s}{\singleton{\sigma}} \in \evalPhiNoArg{\Phi} &
  \mbox{\{Prop.~\ref{prop:glob-sem-comp},\ref{prop:glob-sem-comp-update}\}} \\
  \Leftrightarrow & 
  \sigma \models \judge{\update v=e;s}{\Phi} &
  \mbox{\{Def.~\ref{def:sequent-semantics}\}} \\
  \end{array} $$
  We therefore have that the sequent 
  $\sequent{}{\upl v:= e\upr\judge{s}{\Phi}}$
  is valid if and only if the sequent
  $\sequent{}{\judge{v=e;s}{\Phi}}$
  is valid.
  \qed
  
\paragraph{Rule $\mathsf{Cond}$.}
We shall prove soundness and reversibility of the rule.
Let $\sigma$ be a state. We have, with $\tau=\eval{\update}{\singleton{\sigma}}$:
  $$ \begin{array}{lll}
     & 
  ~~~
  \sigma \models \update (e)  \>\Rightarrow\>
  \sigma \models \judge{\update s; s'}{\Phi} & \\
     & 
  \wedge\
  \sigma \models \update (!e)  \>\Rightarrow\>
  \sigma \models \judge{\update s'}{\Phi} & \\
  \Leftrightarrow~ & 
  ~~~
  \sigma \models \update (e)  \>\Rightarrow\>
  \eval{\update\, s; s'}{\singleton{\sigma}} \in \evalPhiNoArg{\Phi} & \\
     & 
  \wedge\
  \sigma \models \update (!e)  \>\Rightarrow\>
  \eval{\update\, s'}{\singleton{\sigma}} \in \evalPhiNoArg{\Phi} &
  \mbox{\{Def.~\ref{def:sequent-semantics}\}} \\
  \Leftrightarrow~ & 
  ~~~
  \sigma \models \update (e)  \>\Rightarrow\>
  \tau \,\chopSem\,
  \eval{s;s'}{\tau}  \in \evalPhiNoArg{\Phi} ~~~ & \\
     & 
  \wedge\
  \sigma \models \update (!e)  \>\Rightarrow\>
  \tau \,\chopSem\,
  \eval{s'}{\tau} \in \evalPhiNoArg{\Phi} &
  \mbox{\{Prop.~\ref{prop:glob-sem-comp-update}\}} \\
  \Leftrightarrow~ & 
  ~~~
  \valBP{\mathit{last}(\tau)}{}{e} = \trueSem  \>\Rightarrow\>
  \tau \,\chopSem\,
  \eval{s;s'}{\tau}  \in \evalPhiNoArg{\Phi} ~~~ & \\
     & 
  \wedge\
  \valBP{\mathit{last}(\tau)}{}{e} = \falseSem  \>\Rightarrow\>
  \tau \,\chopSem\,
  \eval{s'}{\tau} \in \evalPhiNoArg{\Phi} &
  \mbox{\{Def.~$\sigma \models \update (e)$\}} \\
  \Leftrightarrow & 
  \tau \,\chopSem\,
  \eval{\ifStmt{e}{s};s'}{\tau}
  \in \evalPhiNoArg{\Phi} &
  \mbox{\{Fig.~\ref{fig:local}, Def.~\ref{def:glob-sem}\}} \\
  \Leftrightarrow & 
  \eval{\update \ifStmt{e}{s}; s'}{\singleton{\sigma}} 
  \in \evalPhiNoArg{\Phi} &
  \mbox{\{Prop.~\ref{prop:glob-sem-comp-update}\}} \\
  \Leftrightarrow & 
  \sigma \models \update\judge{\ifStmt{e}{s}; s'}{\Phi} &
  \mbox{\{Def.~\ref{def:sequent-semantics}\}} \\
  \end{array} $$
  \paragraph{Explanation:}
     Given Fig.~\ref{fig:local}, Def.~\ref{def:glob-sem} we have
     \begin{align}
        &\semSeq
        {\semPair{\tau}{\ifStmt{e}{s};s'}}
        {\semPair{\tau}{s;s'}}{*}, \  \text{if} \ \valB{last(\tau)}{e}=\trueSem, \ \text{and}\notag\\
        &\semSeq
        {\semPair{\tau}{\ifStmt{e}{s};s'}}
        {\semPair{\tau}{s'}}{*}, \  \text{if} \ \valB{last(\tau)}{e}=\falseSem \notag
     \end{align}
     Therefore 
     \begin{align}
        &\eval{\ifStmt{e}{s};s'}{\tau}=\eval{s;s'}{\tau} \ \text{if} \ \valB{last(\tau)}{e}=\trueSem, \ \text{and}\notag\\
        &\eval{\ifStmt{e}{s};s'}{\tau}=\eval{s'}{\tau} \  \text{if} \ \valB{last(\tau)}{e}=\falseSem \notag
     \end{align}
  %
  %
  We therefore have that the sequent 
  $\sequent{}{\update\judge{\mbox{\lstinline{if}}\,(e)\,s; s'}{\Phi}}$
  is valid if and only if the sequents
  $\sequent{\update(e)}{\judge{\update s; s'}{\Phi}}$
  and
  $\sequent{\update(!e)}{\judge{\update s'}{\Phi}}$
  are valid.

\paragraph{Rule $\mathsf{Unfold}$.}
We shall prove soundness and reversibility of the rule.
By Tarski's fixed-point theorem for complete lattices~\cite{TarskiPJM55}, 
the semantics 
$\evalPhiS{(\muFml{X (\overline{y})}{\Phi}) (\overline{t})}$ of 
a fixed-point predicate $\muFml{X (\overline{y})}{\Phi}$ 
is indeed a fixed point of the trace predicate transformer
$\lambda F. \lambda \overline{d}. \evalPhi{\Phi}{}
   {\beta [\overline{y} \mapsto \overline{d}]}{\rho [X \mapsto F]}$.
We therefore have the following \emph{fixed-point unfolding} equivalence:
$$ (\muFml{X (\overline{y})}{\Phi}) (\overline{t})
   \>\equiv\>
   \Phi \!\left[(\muFml{X (\overline{y})}{\Phi}) / X,
                \overline{t} / \overline{y}\right] $$
where $\Phi_1 \equiv \Phi_2$ is defined to hold whenever 
$\evalPhiS{\Phi_1} = \evalPhiS{\Phi_2}$
for all $\beta$ and~$\rho$.
The soundness and reversibility of the rule are a direct
consequence of this equivalence.

\paragraph{Rule $\mathsf{ProcedureContract}$.}
%
Because the details are somewhat technical, soundness is only be
sketched here.
We follow the approach taken in~\cite{Oheimb99} to prove
the soundness of a similar rule, but in the context of Hoare logic.
The essence of the approach is to find a suitable notion of validity
of sequents that allows to capture an inductive argument on the
recursive depth of procedure calls. 
In~\cite{Oheimb99}, this is achieved by augmenting the notion of
sequent validity with an explicit parameter~$n$ of that depth,
in turn relying on a modified version of the operational semantics 
that is also parameterised on~$n$ as a bound on the maximal
recursion depth when going from an initial state to a final one.
Here, we follow the same approach, and apply it to traces.

\paragraph{Rule $\mathsf{TrAbs}$.}
We prove soundness of the rule.
%
%
$$
\seqRule{TrAbs}{
\sequent{}{\judge{\mathcal{U}_1}{\Phi_1}}\quad
\sequent{}{\mathcal{U}_1(pre_m(e)})\quad
\seq{\mathbf{C}_m}{\judge{\upP{v:=f_m(\mathcal{U}_1(e))}\mathcal{U}_2}{\Phi_2}}
}{
\sequent{\mathbf{C}_m}{\judge{\mathcal{U}_1\upl v:=m(e)\upr\,\mathcal{U}_2}{\Phi_1\chop\,\Phi_m(e,k)\chop\,\Phi_2}}
}
$$
where $f_m(\cdot)$ is the function computed by $m$.

Assuming the premisses (I)--(III) (from left to right) are valid, we have to show that the conclusion is valid. This means that for all states $\sigma$
$$
\sigma \models (\Gamma \wedge \mathbf{C}_m)\rightarrow \judge{\mathcal{U}_1\upl v:=m(e)\upr\,\mathcal{U}_2}{\Phi_1\chop\,\Phi_m(e,k)\chop\,\Phi_2}
$$
holds. We consider only the non-trivial case where $\sigma\models \Gamma \wedge \mathbf{C}_m$ holds. This means we have to prove that 
$$\eval{\mathcal{U}_1\upl v:=m(e)\upr\,\mathcal{U}_2}{\singleton{\sigma}}\in\evalPhiNoArg{\Phi_1\chop\,\Phi_m(e,k)\chop\,\Phi_2}$$
We can decompose the left side as follows:
$$
\eval{\mathcal{U}_1\upl v:=m(e)\upr\,\mathcal{U}_2}{\singleton{\sigma}}
= \concatTr{\underbrace{\concatTr{\tau}{\eval{v=m(e)}{\tau}}}_{\tau'}}{\tau''},~\text{with}~\eval{\mathcal{U}_1}{\singleton{\sigma}}=\tau~\text{and}~\eval{\mathcal{U}_2}{\tau'}=\tau''
$$
Validity of premise (I) ensures already that
$\tau\in\eval{\Phi_1}{\singleton{\sigma}}$.

\noindent For the middle part, we observe that
\begin{align*}
\eval{v=m(e)}{\tau} & =\overbrace{\callEvP{\last(\tau)}{m,e,k}\chopSem\pushEvP{\last(\tau)}{(m,i)}}^{\bar{\tau}}\chopSem\eval{mb[p\mapsto e];v=\code{res}_{k}}{\bar{\tau}}\\
& =\eval{m(e);}{\tau}\chopSem\eval{v=\code{res}_{k};}{\hat{\tau}}	
\end{align*}
with $\hat{\tau}=\eval{m(e);}{\tau}$.

\noindent By assumption $\sigma\models \mathbf{C}_m$, i.e., $$\sigma\models\forall n,i.(pre_m(n)\rightarrow\judge{m(n)}{\Phi_m(n,i)\chop \stateFml{\code{res}_i\doteq f_m(n)}})$$ and hence, 
$$\sigma\models (pre_m(e_1)\rightarrow\judge{m(e_1)}{\Phi_m(e_1,k)\chop \stateFml{\code{res}_k\doteq f_m(e_1)}})$$
with $\valB{\sigma}{e_1}=\valB{\sigma}{\mathcal{U}_1 e}$ and $e_1$ fresh rigid constant symbol and $i$ instantiated with $k$.
Validity of premise (II) asserts that $$\sigma\models\mathcal{U}_{1}pre_{m}(e) \Leftrightarrow \sigma\models pre_{m}(\mathcal{U}_{1}e)\Leftrightarrow \sigma\models pre_{m}(e_1)\enspace.$$

Consequently (modus ponens), 
\begin{align*}
 & \sigma \models \judge{m(e_1)}{\Phi_m(e_1,k)\chop \stateFml{\code{res}_k\doteq f_m(e_1)}}\\
\Leftrightarrow & \eval{m(e_1)}{\singleton{\sigma}}\in \eval{\Phi_m(e_1,k)\chop \stateFml{\code{res}_k\doteq f_m(e_1)}}{\singleton{\sigma}}\\
\Rightarrow & \eval{m(e_1)}{\singleton{\sigma}}\in \eval{\Phi_m(e_1,k)}{\singleton{\sigma}}
\end{align*}
 
Because there are no side effects from procedure calls on the state,
we have that if for any two states $\sigma,\,\sigma'$
$$
\valB{\sigma}{e}=\valB{\sigma'}{e}\quad\text{and}\quad\valB{\sigma}{k}=\valB{\sigma'}{k}
$$
then
$$
\eval{m(e)}{\singleton{\sigma}}=\eval{m(e)}{\singleton{\sigma'}}\quad\text{and}\quad\eval{\Phi_m(e,k)}{\singleton{\sigma}}=\eval{\Phi_m(e,k)}{\singleton{\sigma'}}\enspace.
$$

Thus we can deduce: $$\eval{m(e_1)}{\singleton{\sigma}}=\eval{m(e)}{\singleton{\last(\tau)}}=\eval{m(e)}{\tau}=\hat{\tau}
\quad\text{, and further,}\quad\hat{\tau}\in\eval{\Phi_m(e,k)}{\tau}\enspace.$$

In summary, we have now that
$\concatTr{\tau}{\hat{\tau}}\in \eval{\Phi_1\chop\Phi_m(e,k)}{\singleton{\sigma}}$
We have not yet considered the whole trace of the procedure call update. Remember:
$$\eval{v=m(e)}{\tau}=\eval{m(e);v=\code{res}_{k}}{\tau}=\eval{m(e);}{\tau}\chopSem\eval{v=\code{res}_{k};}{\hat{\tau}}$$

It remains to show that $$\eval{v=\code{res}_{k};}{\hat{\tau}}\chopSem\eval{\mathcal{U}_2}{\tau'}\in\eval{\Phi_2}{\hat{\tau}}$$

This is a direct consequent of premise (III), the only critical point
being the equality of the value of variable $v$. This follows from the
fact that in the conclusion $v$ has the value computed by the
procedure when called with parameters $\mathcal{U}_1(e)$ which is the
same value to which $f_m(\mathcal{U}_1(e))$ evaluates by definition of
$f_m$.


\qed


\end{proof}

\section{Extended Examples}
\label{app:extended-examples}

\begin{example}[Extended Version of Example \ref{ex:calculus:straight-line}]
  Here is the symbolic execution of the straight-line program ``s =
  \code{if (k!=0) \{r=k-1; r=r+1;\} return r}'' with premise
  $\Gamma:=\code{k}>0$.  For well-formedness, since the program ends
  with a return statement it has to be preceded by a leading update
  $\upP{\startEv(\code{m},\code{k},\code{i}')}\update$. We assume
  $\update$ to be empty.
\begin{prooftree}
  \AxiomC{\LeftLabel{}\text{(symbolic execution finished)}} 
  \UnaryInfC{\RuleL{}$\seq{\code{k}>0}{\judge{\update_1\upP{\code{r}:=\code{k}-1}\upP{\code{r}:=\code{r}+1} \upP{\finishEv(\code{m},\code{r},\code{i}')}\upP{\code{res}_{\code{i}'}:=\code{r}}}{\Phi}}$}
  \UnaryInfC{\RuleL{Assign}$\seq{\code{k}>0}{\judge{\update_1\upP{\code{r}:=\code{k-1}}\upP{\code{r}:=\code{r+1}} \upP{\finishEv(\code{m},\code{r},\code{i}')}\code{res}_{\code{i}'}\code{=r}}{\Phi}}$}
  \UnaryInfC{\RuleL{Return}$\seq{\code{k}>0}{\judge{\update_1\upP{\code{r}:=\code{k}-1}\upP{\code{r}:=\code{r+1}}\mbox{\lstinline|return r|}}{\Phi}}$} 
  \UnaryInfC{\RuleL{Assign}$\seq{\code{k}>0}{\judge{\update_1\upP{\code{r}:=\code{k-1}}\mbox{\lstinline|r=r+1; return r|}}{\Phi}}$} 
  \UnaryInfC{\RuleL{Assign}$\seq{\code{k}>0}{\judge{\update_1\mbox{\lstinline|r=k-1; r=r+1; return r|}}{\Phi}}$} 
  \UnaryInfC{\RuleL{Cond}$\seq{\code{k}>0}{\judge{\underbrace{\upP{\startEv(\code{m},\code{k},\code{i}')}}_{\update_1}\mbox{\lstinline|if (k!=0) \{r=k-1; r=r+1;\} return r|}}{\Phi}}$} 
\end{prooftree}
\end{example}

{
\newcommand{\smallCodeM}{\text{\scriptsize{\code{m}}}}
\begin{example}[Extended Version of Example \ref{ex:procedure-contract-rule}]
  Symbolic execution for procedure \code{m} from Example~\ref{lst:running-example} with contract
   $
\mathbf{C}_{\smallCodeM}=\forall
n,i.(\judge{\code{m}(n)}{\Phi_{\smallCodeM}(n,i)\chop \stateFml{\code{res}_i\doteq n}})$, with $\update=\upP{\startEv(\code{m},\code{n}',\code{i}')}$:
  \begin{prooftree}
    \AxiomC{(symbolic execution of procedure body starts)}
    \UnaryInfC{\RuleL{}$\seq{\code{n}'\geq 0,\,\mathbf{C}_{\smallCodeM}}{\judge{\update\upP{\code{k}':=\code{n}'}\upP{\code{r}':=0} \code{s}[\code{k}/\code{k}',\code{r}/\code{r}']}{\Phi_{\smallCodeM}(\code{n}',\code{i}')}}$}
    \UnaryInfC{\RuleL{VarDecl}$\seq{\code{n}'\geq 0,\,\mathbf{C}_{\smallCodeM}}{\judge{\update\upP{\code{k}':=\code{n}'} \{\code{r};\code{s}[\code{k}/\code{k}']\}}{\Phi_{\smallCodeM}(\code{n}',\code{i}')}}$}
    \UnaryInfC{\RuleL{Assign}$\seq{\code{n}'\geq 0,\,\mathbf{C}_{\smallCodeM}}{\judge{\update \code{k}'=\code{n}';\{\code{r;s}[\code{k}/\code{k}']\}}{\Phi_{\smallCodeM}(\code{n}',\code{i}')}}$}
    \UnaryInfC{$\seq{\code{n}'\geq 0,\,\mathbf{C}_{\smallCodeM}}{\judge{inline(\code{m},\code{n}',\code{i}')}{\Phi_{\smallCodeM}(\code{n}',\code{i}')}}$}
    \UnaryInfC{\RuleL{ProcedureContract}$\seq{}{\mathbf{C}_{\smallCodeM}}$}
  \end{prooftree}
  
\end{example}

}\begin{example}[Extended Version of Example
  \ref{ex:trace-abstraction}] We provide some of the needed update
  simplification rules:
$$\begin{array}{l}

\seqRule{elimUpdate_1}{
\seq{\Gamma}{\judge{\update}{\Phi},~\Delta}\qquad\seq{\Gamma}{\update\upP{x:=e}\varphi,~\Delta}
}{
\seq{\Gamma}{\judge{\update\upP{x:=e}}{\Phi\concat\stateFml{\varphi}},~\Delta}
}

\\[2em]

\seqRule{elimUpdate_2}{
\seq{\Gamma}{\judge{\update}{\Phi},~\Delta}\qquad
\seq{\Gamma}{(\update m)\doteq m' \wedge (\update e)\doteq e' \wedge (\update i)\doteq i',~\Delta}
}{
\seq{\Gamma}{\judge{\update\upP{\finishEv(m,e,i)}}{\Phi\chop\finishEv(m',e',i')},~\Delta}
}

\\[2em]

\seqRule{subsumeUpdates_1}{
\seq{\Gamma}{\judge{\update_1}{\Phi},~\Delta}
}{
\seq{\Gamma}{\judge{\update_1\update_2}{\Phi~\finiteNoM{m}},~\Delta}
}\qquad
    \begin{minipage}{.3\linewidth}
      if $\update_2$ does not contain any update involving $m$
    \end{minipage}

\end{array}$$

\end{example}



\end{document}